\documentclass[11pt]{article}
\usepackage[english]{babel}
\usepackage[utf8]{inputenc}
\usepackage{fancyhdr}
\usepackage{amsmath,amssymb,amsfonts,amsthm,amscd}

\usepackage{graphicx}
\usepackage{psfrag}

\newcommand{\bm}[1]{\mbox{\boldmath $#1$}}

\newcommand{\mnotex}[1]
{\protect{\stepcounter{mnotecount}}$^{\mbox{\footnotesize $\bullet$\themnotecount}}$ 
\marginpar{
\raggedright\tiny\em
$\!\!\!\!\!\!\,\bullet $\themnotecount: #1} }

\theoremstyle{plain}

\newtheorem{Tma}{Theorem}
\newtheorem{Prop}{Proposition}
\newtheorem{Lem}{Lemma}
\newtheorem{Cor}{Corollary}

\theoremstyle{definition}
\newtheorem{Def}{Definition}

\theoremstyle{remark}
\newtheorem{remark}{Remark}

\def\Eing{\mbox{Ein}^g}
\def\Scal{\mbox{Scal}}
\def\Ricg{\mbox{Ric}^g}
\def\Riem{\mbox{Riem}}
\def\Riemg{\mbox{Riem}^g}
\def\tr{\mbox{tr}}

\def\X{\mathfrak{X}}
\def\M{{\mathcal M}}
\def\gMm{g}
\def\F{{\mathcal F}}
\def\ellS{\ell_S}
\def\la{\langle}
\def\ra{\rangle}
\def\nablaS{\nabla^S}

\def\s{s}
\def\q{{\mathring{q}}}
\def\qh{\hat{q}}
\def\stilde{\tilde{\s}}
\def\sone{\s^{(1)}_{\ell}{}}
\def\sonep{\s^{(1)}_{\ell'}{}}
\def\gammatilde{\tilde{\gamma}}
\def\thltilde{\tilde{\theta}_{\ell}}
\def\thlo{\theta_{\ell}^{(0)}}
\def\thlone{\theta_{\ell}^{(1)}}
\def\thkone{\theta_{k}^{(1)}}
\def\A{\thlone}

\def\K{\mathcal K}

\def\be{\begin{equation*}}
\def\en{\end{equation*}}
\def\grad{\mbox{grad \,}}
\def\Hess{\mbox{Hess}}
\def\TX{X^{\prime}}
\def\TY{Y^{\prime}}

\def\Ft{F}
\def\Dh{\hat{D}}
\def\B{\theta_{\ell'}^{(1)}}
\def\I{\mathcal I}

\def \volunitdos{\bm{\eta_{\q}}}

\def\s {s}
\def\r {\overline{r}}
\def\rt {\overline{r}'}
\def\h {\mathring{h}}
\def\hh {\hat{h}}
\def\c {{\thkone}}
\def\a {{\thlone}}
\def\ch {\hat{\theta}_k^{(1)}}
\def\ah {\hat{\theta}_{\ell}^{(1)}} 

\def\defi{:=}
\def\esf {\mathring{\nabla}}

\def \rb {\overline{r}}
\def \rbp {\overline{r}'}

\def \zet {\overline{z}}
\def \fun {\xi}
\def \ordone { o_1(1)}
\def \ordtwo { o_2^X(1) }
\def \ordthree {o_1^X(\rb^{-1})}

\def\JournalPrep#1#2#3{#1, ``#2'', #3.}
\def\Journal#1#2#3#4#5#6{#1, ``#2'', {\em #3} {\bf #4}, #5 (#6).}

\def\JDG{\em J. Diff. Geom.}
\def\CQG{\em Class. Quantum Grav.}
\def\JPA{\em J. Phys. A: Math. Gen.}

\def\GRG{\em Gen. Rel. Grav.}

\def\JMP{\em J. Math. Phys.}

\def\CMP{\em Commun. Math. Phys.}

\def\PRL{\em Phys. Rev. Lett.}

\def\ANYAS{\em Ann. N. Y. Acad. Sci.}

\def\ATMP{\em Adv. Theor. Math. Phys.}
\def\PRSLA{\em Proc. Roy. Soc. London A.}

\begin{document}

\title{The asymptotic behaviour of the Hawking energy along null asymptotically flat hypersurfaces}

\author{Marc Mars$^1$ and Alberto Soria$^2$ \\
Facultad de Ciencias, Universidad de Salamanca,\\
Plaza de la Merced s/n, 37008 Salamanca, Spain \\
$^1$ \,marc@usal.es,  $^2$ \,asoriam@usal.es, }

\maketitle

\begin{abstract}

In this work we obtain the limit of the Hawking energy
of a large class of foliations along general
null hypersurfaces $\Omega$
satisfying a weak notion of asymptotic
flatness. The foliations
are not required to be either geodesic or approaching large spheres
at infinity. The limit is obtained in terms of a reference 
background geodesic foliation approaching large spheres and a positive
function, constant along the null generators on $\Omega$,
which describes the relation
between the two foliations at infinity. The integrand in the limit expression
has interesting covariance and invariance properties with respect to change
of background foliation. The standard result that the Hawking energy
tends to the Bondi energy under suitable circumstances is recovered in this framework. 
 \end{abstract}

\section{Introduction}

For any closed spacelike surface $S$ with spherical topology embedded
in a four dimensional spacetime, the Hawking energy is defined by 
\begin{equation} 
m_H(S)=\sqrt{\frac{|S|}{16\pi}} \left (
1-\frac{1}{16\pi}\int_{S}\vec{H}^2\eta_{S} \right ),
\end{equation} where $\vec{H}$ is the mean curvature of $S$ and $|S|$
is the area of $S$.

The Hawking energy was introduced by Hawking \cite{Hawking1968} and
under certain circumstances it gives a measure 
of the total energy contained in the surface $S$. However, it is well-known
that this is not always the case. For instance, in the Minkowski spacetime
any surface embedded in
a spacelike hyperplane has negative Hawking energy, except for
round spheres where it vanishes, while
there are spacelike surfaces embedded in the time-cylinder over a
two-sphere with positive Hawking energy. For
surfaces embedded in the past null cone of a point (also in Minkowski)
the Hawking energy turns out to be identically zero. In more
general spacetimes, the small sphere limit of the Hawking energy
has been studied by Horowitz and Schmidt \cite{HorowitzSchmidt1982} 
who found that, for suitably constructed surfaces embedded
in the future null cone of a point $p$, the leading term
of the Hawking energy agrees with the energy-density 
at the vertex or, if the spacetime is vacuum at $p$, 
by a suitable time-component of the Bel-Robinson
tensor. Hence the Hawking energy enjoys interesting
positivity properties in this limit. 

In the opposite regime of very large spheres in an asymptotically flat 
spacetime the Hawking energy also has interesting properties,
at it is known that
for suitably round spheres at infinity, the Hawking energy
approaches the ADM or Bondi energies. More precisely,
in the asymptotically flat spacelike context, the Hawking energy
of the surfaces of constant coordinate radius $r$ in the asymptotic region 
has a limit when $r \rightarrow \infty$ which agrees with 
the ADM energy of the hypersurface \cite{Hawking1968}. In the
asymptotically hyperboloidal case, the approximate Bondi spheres
have the same property, the limit now being the Bondi energy (see
\cite{ChruscielJezierskiLeski2004} for details). When the surfaces are embedded
in a null hypersurface intersecting null infinity on a cross section, the
limit of the Hawking energy is again the Bondi energy provided the
surfaces approach large spheres in the sense that 
(see \cite{PenroseRindler}, \cite{Bartnik2004} and Sauter's
Ph.D. thesis \cite{Sauter2008} where this result is explicitly quoted):
\begin{equation}
\label{largespherescondition}
\K_\infty\defi\underset{r \to \infty}{\lim}\frac{|S_r|}{4\pi}\K(r)=1,
\end{equation}
where $\{S_r\}$ is the collection of surfaces along which the limit is taken,
$|S_r|$ is the area of $S_r$ and $\K(r)$ its Gauss curvature. 
Since $\K_\infty$ is the Gauss curvature of the surface at infinity after
a suitable rescaling, the condition above states that the surfaces
approach large round spheres at infinity, in agreement with the behaviour
in the asymptotically Euclidean and hyperboloidal cases. 

It turns out, however, that understanding the behaviour of the
Hawking energy at infinity when
the condition of round spheres is not imposed is much more subtle. The aim
of this paper is to carry out such an analysis for surfaces embedded
in a asymptotically flat null hypersurface (we give below the 
precise definition). This problem
is interesting for several reasons. First of all, it is relevant
in order to help clarifying the physical meaning of the Hawking energy, which,
as already said, is related to an energy in some circumstances but not
in others. From a more practical point of view, the  Hawking energy has become
a very valuable tool for various problems in geometric analysis. The 
underlying reason is that the Hawking energy enjoys interesting monotonicity
properties for specific flows of surfaces. In order to become
truly useful, this monotonicity needs to
be complemented with a good behaviour of the
Hawking energy at infinity, so that its asymptotic value
can be related to the ADM (or Bondi) energies of the spacetime. Whenever
the flow can be proved to approach large round spheres, the results above
suffice, but often this is not the case and understanding the behaviour
of the Hawking energy at infinity under general
circumstances becomes a useful piece of information.

To be more specific, the Hawking energy has played a fundamental role in
the proof by Huisken \& Ilmanen \cite{HuiskenIlmanen2001}
of the Penrose inequality 
in the time-symmetric, asymptotically 
euclidean case in four spacetime dimensions. The key fact behind their proof 
is a monotonicity formula for the Hawking energy under the inverse
mean curvature flow, first discovered by Geroch \cite{Geroch1973}
and extended by Huisken and Ilmanen to a suitably weak setting that
guarantees global  existence of the flow. On the
horizon $S_H$ (a connected outermost minimal surface in this case), the
Hawking energy
agrees with $\sqrt{|S_H|/16 \pi}$. At infinity, the authors were able to prove
that the flow makes the surfaces sufficiently round so as to guarantee 
that the limit of the Hawking energy is not larger than the ADM mass, thus establishing the Penrose inequality $M_{\mbox{\tiny ADM}} \geq \sqrt{|S_H|/16 \pi}$. 

Monotonicity of the Hawking energy has been studied in 
various contexts, both as codimension one flows within spacelike
\cite{MalecMarsSimon2002} or null hypersurfaces \cite{Sauter2008} or as codimension-two
flows in a spacetime setting \cite{Hayward1994,BrayHaywardMarsSimon2006}
where the notion of uniformly expanding flows where introduced and sufficient
conditions for mononicity of the Hawking energy were found. Spacetimes flows
under which the Hawking energy is monotonic have received renewed interest
recently \cite{BrayJauregui2014},
where additional sufficient conditions for monotonicity have been found 
and the role of so-called time flat surfaces
has been emphasized (see \cite{BrayJaureguiMars2014} for the relationship
with the previous spacetime monotonicity results). All these results show that 
monotonicity of the Hawking energy is a versatile property which can be
accommodated to many circumstances. 

However, the limit at
infinity of the flows turns out to be much more problematic. 
This was first realized by A. Neves 
\cite{Neves2010} who studied inverse 
mean curvature flows in Riemannian,
asymptotically hyperbolic 3-dimensional manifolds
with scalar curvature bounded below by a negative constant. It turns
out that the flow does not guarantee convergence of the surfaces
to sufficiently round spheres, which, in general, prevents 
comparison of the limit of the Hawking energy and the total
mass of the hypersurface. A similar difficulty is faced for flows along null
hypersurfaces \cite{Sauter2008}.

It is therefore interesting to know, in as much generality as possible,
what is the limit of the Hawking energy at infinity without assuming that
the surfaces $\{S_r\}$ approach large spheres at infinity. The null case
is particularly interesting because it allows for a very neat description
of spacelike surfaces embedded in the null hypersurfaces as graphs with 
respect to a background foliation that can be chosen conveniently.
We exploit this fact in order to obtain an explicit and simple
expression for the limit of the Hawking energy at
infinity for a very general flow of spacelike surfaces.  Our main result
is as follows (see Section \ref{Sect2} for definitions)
\begin{Tma}
\label{main}
Let $\Omega$ be a past asymptotically flat null hypersurface. Let $\{ S_r\}$
be a foliation defined as the level sets of a function $r: \Omega \rightarrow
\mathbb{R}$ satisfying $k(r)=-1$ where $k$ is future, tangent to the null
generators of $\Omega$ and geodesic. Let $\ell$ be the unique
null vector orthogonal to $S_r$ and satisfying $\la \ell,k\ra = -2$. 
Assume that $\{S_r\}$ approaches large spheres with round limit metric $\q$.
Consider
any flow of surfaces $\{S_{r^{\star}}\} $, $r^{\star} = \mbox{const.}$
defined by
\begin{equation}
r= \phi r^{\star} +\tau+\fun^{\star},
\end{equation}
where $\phi, \tau$ are Lie constant along $k$ with $\phi  > 0$ everywhere
and  $\fun$ satisfies $\fun=\ordone\cap\ordtwo$, $k(\fun)=\ordthree$.
Then, the limit of the Hawking energy along $\{ S_{r^{\star}} \}$ is
\begin{equation}
\label{hawkingmasslimittwoone}
\underset{r^{\star} \to \infty}{\lim}m_H(S_{r^{\star}})=\frac{1}{8\pi\sqrt{16\pi}}\left(\sqrt{\int_{\mathbb{S}^2} \phi^2 \volunitdos}\right)\int_{\mathbb{S}^2}\left(\triangle_{\q}\c-(\c+\a)-4 \mbox{div}_{\q} (\sone) \right) \frac{1}{\phi} \volunitdos,
\end{equation} 
where 
\begin{align*}
\c = \lim_{r \rightarrow \infty} \left ( \theta_k(S_r) r^2 + 2 r \right ),
\quad 
\a = \lim_{r \rightarrow \infty} \left ( \theta_{\ell}(S_r) r^2 - 2 r \right ),
\quad 
\sone
 = 
\lim_{r \rightarrow \infty} \left ( r \s_{\ell} \right )
\end{align*}
and $\theta_{k}(S_r)$, $\theta_{\ell}(S_r)$ are the null expansions of $S_r$ 
and $\s_{\ell}(X) := \frac{1}{2} \la \ell,\nabla_X k\ra$, $X \in \mathfrak{X}(S_r)$
is the connection of the normal bundle of $S_r$.
\end{Tma}

For definiteness, we work with past null hypersurfaces in this paper. The
changes required to deal with future null hypersurfaces
are indicated in Remark \ref{futurenull} below.

In addition to this theorem we also find an interesting covariance property
of the integrand in (\ref{hawkingmasslimittwoone}) under changes of background 
foliation. This is
part of the content of Theorems \ref{Tmalargespheres}, \ref{Tmatau}, and \ref{Tmaq} below, and which may be of independent interest.

We conclude by noting that
the Hawking energy is also useful as a tool to control the Bondi mass
of a spacetime. For null hypersurfaces this idea has been exploited by
Alexakis and Shao \cite{AlexakisShao2014}, who prove
bounds for the Bondi energy-momentum for vacuum, null hypersurfaces
close to the shear-free outgoing null cones in Schwarzschild spacetime.
According to these authors obtaining the limit of the Hawking energy
is particularly difficult, fundamentally  due to the lack of a 
neat formula for the change of the mass aspect function 
under change of geodesic foliation. 
Recall that the mass aspect function \cite{C,CK}
is  a specific geometric quantity $\mu_S$ defined on surfaces with the property
that 
\begin{align*}
m_H(S)=\sqrt{\frac{|S|}{16\pi}} \int \mu_S \bm{\eta}_S.
\end{align*}
Theorems  \ref{Tmalargespheres}, \ref{Tmatau} and \ref{Tmaq}
do provide a simple transformation formula at infinity
for an integrand of the Hawking energy for foliations approaching large
spheres, so a more general application of this result may well have
interesting applications also in the context of Alexakis and Shao's work.

As a more general application of our results,
flows along null hypersurfaces are likely to play
a role in many attemps to prove the Penrose inequality at null infinity,
cf. \cite{LudvigsenVickers1983}, \cite{Bergqvist1997},
\cite{Mars2009}, \cite{MarsSoria2012}, \cite{BrendleWang}, \cite{MarsSoria2013}.
The general methods developed in this paper to deal
with general foliations in null hypersurfaces are thus likely to be useful
in that context as well.

The paper is organized as follows.
In Section \ref{Setup} we recall the geometry of graphs over a background
foliation in a null hypersurface. For the sake of generality,
we work here in arbitrary dimension and with general foliations transverse
to the null generator, even 
though for the rest of the paper only dimension four and geodesic
background foliations are used.
In Section \ref{Sect2} we introduce our definition of asymptotically flat null
hypersurface and consider some simple consequences. We
also introduce the notion of energy flux decay condition which imposes
additional decay on some components of the Einstein tensor and which
is used later on in Section \ref{Sect6}.
In Section \ref{Sect3}
the definition of approaching large spheres is introduced. 
Section \ref{Sect4} is devoted to find the limit of the Hawking energy
for geodesic foliations. Given a background foliation approaching large
spheres we consider first what happens for other geodesic foliation
starting at the same initial surface (Theorem \ref{TmaHawkinglimit}) 
and then the case when the foliation has the same null generator but
starts on a different surface (Theorem \ref{Tmatau}).
 We find, in particular, an interesting covariance
property of the integrand arising in the limit expression of the Hawking
mass under changes of background foliation (Theorems \ref{Tmalargespheres}
and \ref{Tmatau}). The case of non-geodesic foliations is treated in 
Section \ref{Sect5}. We devote Section  \ref{Sect6}  to recover
the well-known result that the Hawking energy tends to the Bondi energy
when the foliation approaches large spheres. This involves recalling
known properties of the relationship between the conformal group of the sphere
and the Lorentz group in Minkowski, with a subtlety that arises
when the hypersurfaces extend to past null infinity. We also compare
our results with the analysis of the limit of the Hawking energy in the so-called
null quasi-spherical gauge by Bartnik \cite{Bartnik2004}.

\section{Setup}
\label{Setup}
Let $(\M,\gMm)$ be a time-oriented spacetime of dimension $n+1$,
$n \geq 3$.
We consider a smooth, connected, null hypersurface $\Omega$  embedded
in $(\M,\gMm)$. Let $k$ be a smooth, nowhere zero, future directed
null vector field tangent to $\Omega$ (i.e. a null generator).
Since the integral curves of $k$ are geodesics, there
exists $Q_{k} \in \F(\Omega)$ such that $\nabla_k k = Q_k k$, where
$\nabla$ is the Levi-Civita covariant derivative of $(\M,\gMm)$. We make
the assumption that there is an embedded spacelike connected hypersurface
$S_0$ in $\Omega$ (with embedding $\Phi_0$)
such that any integral curve of $-k$ intersects $S_0$ precisely once.
This implies the existence of a smooth map $\pi: \Omega \longrightarrow
S_0$ (we identify $S_0$ with its image, the meaning being clear from the
context) which sends $p\in \Omega$ to the intersection of the integral
curve $\gamma^k_p$  of $-k$ passing through $p$ with $S_0$. The map $\pi$ is
clearly a submersion. We choose the parameter $\lambda$ of
the curve $\gamma^k_p$ so that $\gamma^k_p(0) = p$.

Given $k$,  $S_0$ and a constant $r_0$, a scalar function $r \in \F(\Omega)$ is defined
by $k(r)=-1$ and $r(p) = r_0$ for all $p \in S_0$. Let $(r_{-}(p),r_{+}(p))$ be the range of the function $r$ restricted to
the curve $\gamma^k_p$. We also assume that 
the open interval $(R_{-} := \sup_{S_{0}} r_{-}, R_{+} := \inf_{S_{0}} r_{+})$
is non-empty.  The function $r$ having 
nowhere zero
gradient, the level sets $S_{r_1} = \{ r = r_1 \}$ 
are either empty or smooth, embedded (not necessarily connected) 
hypersurfaces.
The collection $\{ S_r \}$ is a foliation of $\Omega$.
For $r_1 \in (R_{-},R_{+})$ the hypersurfaces $S_{r_1}$ are in fact connected
and diffeomorphic to $S_0$. 

At any $p \in \Omega$ let $\ell |_{p} \in T_p \M$ be the
unique null vector field satisfying $\la k,\ell \ra |_p = -2$ and
$\la \ell, X \ra |_p =0$ for any $X \in T_p S_{r(p)}$. $S_r$ is endowed
with an induced metric $\gamma_{S_r}$, with two null second fundamental forms
$K^{k}(X,Y) := \la \nabla_X k,Y \ra$, $K^{\ell}(X,Y) := \la \nabla_X \ell, Y \ra$
$X,Y \in \X(S_{r})$ and with a normal bundle connection one-form
$\s_{\ell}(X): = - \frac{1}{2} \la \nabla_X \ell, k \ra$. 

In order to obtain the limit of the Hawking energy as described
in the introduction, we need to relate the geometry of
different spacelike surfaces embedded in $\Omega$. Consider
a spacelike embedded hypersurface $S$ in $\Omega$ with
embedding $\Phi: S \longrightarrow \Omega$ and
let $p \in S$. This hypersurface is uniquely defined
by a diffemorphism $\Psi: S \longrightarrow \Psi(S) \subset S_0$
and a function $F \in \F(S)$ as follows. For all $p \in S$ define
$F(p) = r(\Phi(p))$ and $\Psi(p) = (\pi \circ \Phi)(p)$.
Conversely, a function $F \in \F(S)$ with image in $(r_{-}(p),r_{+}(p))$ and a
diffeomorphism $\Psi$ as above
defines an embedding 
\begin{align*}
\Phi: S & \longrightarrow \Omega \\
p & \longrightarrow \gamma^k_{\Psi(p)}(\lambda = F(p)).
\end{align*}
We want to relate the intrinsic and extrinsic geometry
 of $S$ at $p$ with the geometry of 
the surface $S_{r=F(p)}$. Since this is all local we can assume
$\Psi(S) = S_0$, which makes the presentation simpler. We
extend $F$ to a function on $\Omega$ defined by $F(q) =
F((\Psi^{-1} \circ \pi)(q))$. We keep the same
symbol for the extension. 
It is clear
that $k (F)=0$.  For  the extrinsic geometry of $S$ (we again
identify $S$ with its image) we define at
$p \in S$, the null normal $\ellS|_p$ by the conditions
$\la \ellS,k\ra |_p = -2$ and $\la \ellS,X \ra =0$ for all 
$X\in T_p S$. The null second fundamental forms
$K^{\ellS}$, $K^{k}$ and the normal connection one-form $\s_{\ellS}$ are defined
similarly as before. The following Proposition
is known (see e.g \cite{Sauter2008}) 
when the background foliation $\{S _r \}$ is geodesic (i.e. $Q^{k} =0$).
Although this is the situation we will require later, we include
for the sake of generality the non-geodesic case as well. 
Our proof also follows a somewhat different approach.

\begin{Prop}
\label{RelGraph}
Let $p \in S$, then the map
\begin{align*}
T_F : T_p S_{r=F(p)}  & \longrightarrow T_p S\\
      X & \longrightarrow \TX := X - X(F)  k
\end{align*}
is a well-defined isomorphism. The induced metric $\gamma_{S}$,
null second fundamental forms $K^{\ellS}$, $K^{k}$ and
normal bundle connection $\s_{\ellS}$ of $S$ are given by
\begin{align}
\gamma_{S}|_p (\TX,\TY)  = &
\gamma  (X,Y), \label{metric}  \\
K^{k}(\TX,\TY) |_p  = &  K^{k}(X,Y) \label{Kk1} \\
\s_{\ellS} (\TX) |_p  = &  
\s_{\ell}(X)  - K^k(X, \grad F) + X(F) Q_{k}
\label{sprime}
\\
K^{\ell_S}(\TX,\TY) |_p = & 
K^{\ell}(X,Y) + |DF|^2 K^{k} (X,Y) 
+ 2 X(F) \s_{\ell} (Y)  +2  Y(F) \s_{\ell} (X) \nonumber \\
& - 2  X(F) K^k(X, \grad F)  -2  Y(F) K^k(Y, \grad F)  - 2 \Hess \, F
+ 2  Q_{k} X(F) Y(F) 
\label{Kell}
\end{align}
where $\gamma, K^k$, $K^{\ell}$, $\s_{\ell}$, $\grad$, $\Hess$ 
and $|DF|^2 = \la \grad \, F , \grad \, F \ra$
refer
to the surface $S_{r=F(p)}$  and are evaluated at $p$.
\end{Prop}

The proof is based on the following simple identity that
may be useful in other contexts.
\begin{Lem}
\label{deformation}
Let $S$ be an embedded spacelike surface with embedding 
$\Psi: S \longrightarrow \M$. Select a pair of null normal vector fields
$\{ k, \ell\}$ along $S$ satisfying $\la k, \ell \ra = -2$.
For any vector field $\xi$ on a spacetime neighbourhood of
 $S$ write its deformation tensor as $\pounds_{\xi} g:=a_{\xi} $.
Then
\begin{align*}
a_{\xi}(X,Y) & = - \la \xi, \ell \ra K^{k} (X,Y) 
- \la \xi, k \ra K^{\ell} (X,Y) + 
\left ( \nablaS_{X} \bm{\xi^{\parallel}} \right )(Y) +
\left ( \nablaS_{Y} \bm{\xi^{\parallel}} \right )(X), 
\quad \quad X,Y \in \X(S)
\end{align*}
where $\nablaS$ is the Levi-Civita covariant derivative of $S$,
$\bm{\xi}^{\parallel} := \Psi^{\star}(\bm{\xi})$ and
$\bm{\xi} := g(\xi,\cdot)$.
\end{Lem}

\begin{remark}
Given $\xi$ merely along $S$, this result can be applied
to any extension of $\xi$ to a neighbourhood of $S$, the result
being independent of the extension.
\end{remark}

\vspace{3mm}

\begin{proof} Decompose $\xi$ in tangential
and normal parts $\xi = \xi^{\perp} + \xi^{\parallel}$ and
$\xi^{\perp}$ in the null basis $\{k,\ell\}$
$\xi = - \frac{1}{2} \la \xi,\ell \ra k 
- \frac{1}{2} \la \xi,k \ra \ell  + \xi^{\parallel}$
so that
\begin{align*}
a_{\xi} (X,Y) & = 
\la \nabla_X \xi, Y \ra  + \la \nabla_Y \xi, X \ra \\
& = 
- \la \xi,\ell \ra K^{k}(X,Y)
- \la \xi,k \ra K^{\ell}(X,Y)
+ \la \nabla_X \xi^{\parallel}, Y \ra  + \la \nabla_Y \xi^{\parallel}, X \ra \\
& = 
- \la \xi,\ell \ra K^{k}(X,Y)
- \la \xi,k \ra K^{\ell}(X,Y)
+ \la \nablaS_X \xi^{\parallel}, Y \ra  +\la \nablaS_Y \xi^{\parallel}, X \ra \\
& = 
- \la \xi,\ell \ra K^{k}(X,Y)
- \la \xi,k \ra K^{\ell}(X,Y)
+ (\nablaS_X \bm{\xi^{\parallel}})(Y) 
 + (\nablaS_Y \bm{\xi^{\parallel}})(X). 
\end{align*}
\end{proof}

\noindent {\it Proof of the Proposition.} $T_F$ is well-defined provided
$X - X(F)  k$ is tangent to $S$. This follows because $S$ is defined by  $r- F =0$ (note that $d(r-F) \neq 0$
everywhere) and 
${\mathcal L}_{ X - X(F)  k} (r - F) =  - X(F) - X(F) k(r) = 0$. 
$T_F$ is obviously
injective, hence an isomorphism. Properties (\ref{metric}) and
(\ref{Kk1}) are immediate (and well-known). For the remaining 
parts we note the decomposition
\begin{equation}
\ellS |_p  = \ell + |DF|^2  k 
- 2 \, \grad  F \, |_p, \quad \quad \quad p \in S,
\label{ellS}
\end{equation}
which holds because the right-hand side is null, satisfies
$\la \ellS, k \ra = -2$ and is ortohogonal to $\TX=T_F(X)$, for all
$X \in T_pS$.
To show (\ref{sprime}) we compute
\begin{align*}
\s_{\ellS}(\TX) & = - \frac{1}{2} \la \nabla_{\TX} \ellS,k \ra = 
\frac{1}{2} \la \nabla_{\TX}  k, \ellS \ra  = \frac{1}{2} \la \nabla_X k - X(F) \nabla_k k, \ellS \ra  \\
& = \s_{\ell}(X)  - K^k(X, \grad F) + X(F) Q_{k}. 
\end{align*}
For the null extrinsic curvature $K^{\ellS}$ we use Lemma \ref{deformation}.
First observe that the right-hand side of (\ref{ellS}) makes sense for all
$p \in \Omega$, so
it defines an extension of $\ellS$ which remains null and satisfying
$\la \ellS,k \ra = -2$. Extend also $Y \in T_p S_{r=F(p)}$ to a neighbourhood 
under the condition that remains tangent to the foliation $\{S_r\}$. This
induces an extension of $\TY$ which remains orthogonal to
$\ellS$. Note 
\be
[k ,\TY] = [k,Y] - k(Y(F)) k = [k,Y] - [k,Y] (F) = ([k,Y])^{\prime}
\en
which shows that $[k,\TY]$ is tangent to $S$ at $p$ (and we used that
$[k,Y]|_p$ is tangent to $S_{r=F(p)}$). We apply Lemma \ref{deformation}
on the surface  $S_{r= F(p)}$ and to the vector field $\ellS$. 
Concerning the deformation tensor
\begin{align*}
\la \nabla_{X} \ellS, Y \ra & =
\la \nabla_{\TX +  X(F) k} \ellS, \TY + Y(F) k \ra \\
& =
\la \nabla_{\TX} \ellS,\TY \ra + X(F) \la \nabla_k \ellS, \TY \ra
+ Y(F) \la \nabla_{\TX} \ellS, k \ra 
+ X(F) Y(F) \la \nabla_{k} \ellS, k \ra  \\
& = 
K^{\ellS}(\TX,\TY) - X(F) \la \nabla_k \TY,\ellS \ra
- 2 Y(F) \s_{\ellS} (\TX) - X(F) Y(F) \la \nabla_k k, \ellS \ra \\
& = 
K^{\ellS}(\TX,\TY) - X(F) \la [k,\TY] + \nabla_{\TY} k,\ellS \ra
- 2 Y(F) \s_{\ellS} (\TX) + 2 Q_{k} X(F) Y(F) \\
& =K^{\ellS}(\TX,\TY) - 2 X(F) \s_{\ellS} (\TY) 
- 2 Y(F) \s_{\ellS} (\TX) + 2 Q_{k} X(F) Y(F).
\end{align*}
where in the last equality we used that $[k,\TY] |_p $ is tangent
to $S$. Hence
\begin{align}
a_{\ellS} (X,Y) =  &
2 K^{\ellS}(\TX,\TY) - 4 X(F) \s_{\ellS} (\TY) 
- 4 Y(F) \s_{\ellS} (\TX) + 4 Q_{k} X(F) Y(F)  \nonumber \\
= &  
2 K^{\ellS}(\TX,\TY) - 4 X(F) \s_{\ell} (Y)
- 4 Y(F) \s_{\ell} (X) + 4 X(F) K^k(X, \grad F) \nonumber \\
& + 4 Y(F) K^k(Y, \grad F)  - 4 Q_{k} X(F) Y(F) \label{Lem1}
\end{align}
after using (\ref{sprime}) in the second equality. Now
Lemma \ref{deformation} gives
\begin{equation}
a_{\ellS}(X,Y) = 
2 |DF|^2 K^{k} (X,Y)  + 2 K^{\ell}(X,Y) - 4 \, \Hess \, F.
\label{Lem2}
\end{equation}
Solving for $K^{\ellS}(\TX,\TY)$ in (\ref{Lem1}) and (\ref{Lem2}) yields the result.

\vspace{3mm}

The following corollary is a trivial consequence of how
the null second fundamental forms and the normal bundle connection
transforms under a boost in $\{\ell_S, k \}$
\begin{Cor}
\label{boostcor}
Let $S$ as before and for all $p \in S$ let $k' |_p = \alpha(p) k|_p$
and $\ell_S'|_p = \frac{1}{\alpha(p)} \ell_S$, where
$\alpha: S \mapsto \mathbb{R}$ is a smooth positive function. Then
\begin{align}
\s_{\ell_S^{\prime}} (\TX)   = &  
\label{oneformalpha}
\s_{\ell}(X)  - K^k(X, \grad F) + X(F) Q_{k} - \frac{1}{\alpha} X(\alpha) 
 \\
K^{k^{\prime}} (\TX,\TY)   = & \alpha K^k(X,Y)  \nonumber \\
K^{\ell_S^{\prime}}(\TX,\TY) |_p = & 
\frac{1}{\alpha} \left ( K^{\ell}(X,Y) + |DF|^2 K^{k} (X,Y)  
+ 2 X(F) \s_{\ell} (Y) +2  Y(F) \s_{\ell} (X)  
- 2 \Hess \, F \right . \nonumber \\
& \left . 
\hspace{5mm} - 2  X(F) K^k(X, \grad F) 
-2  Y(F) K^k(Y, \grad F)  
+ 2  Q_{k} X(F) Y(F)  \frac{}{}  \right ). \nonumber
 \end{align}
\end{Cor}

The trace of $K^k$ and $K^{\ell}$ on $S_r$ with the induced metric
define the null expansions of $S_r$ and are denoted respectively as
$\theta_k$ and $\theta_{\ell}$. The relationship between the null expansions
$\theta_k$, $\theta_{\ellS}$ of a graph $S$ with the corresponding
ones at the level set $S_{r=F(p)}$ follow from 
Proposition \ref{RelGraph}.

\begin{Cor}
\label{RelExpansions}
Let $S$, $k'$ and $\ell_S'$ as in Corollary \ref{boostcor}. The null
expansions $\theta_{k'}$ and $\theta_{\ell'}$ at $p \in S$ and the
null expansions $\theta_k$, $\theta_{\ell}$ of $S_{r=F(p)}$ at $p$
are related by
\begin{align}
\theta_{k'} & = \alpha \theta_k \\
\theta_{\ell_S'} & = \frac{1}{\alpha} \left ( \theta_{\ell}
+ |DF|^2 \theta_k + 4 \s_{\ell} (\grad F) - 4 K^{k} ( \grad F,
\grad F) - 2 \Delta F + 2 Q_k |DF|^2 \right )
\label{Relthl}
\end{align}
where $\Delta F$ is the Laplacian of $S_r$ with the induced metric.
\end{Cor}

Another useful identity that will play a role later is the evolution equation
for the connection of the normal bundle. For geodesic flows this identity is
known, see e.g. \cite{ChruscielPaetz}. Although again this case is all we
shall need in this paper, we state and prove the result in full generality (i.e.
for arbitrary $k$). 
\begin{Prop}
\label{evols}
With the same notation as above,
let $X \in \X(\Omega)$ be a vector field satisfying $[k,X] =0$ and
tangent to $S_0$. Then 
\begin{align}
k(s_{\ell}(X)) = - X(Q_k) - s_{\ell}(X) \theta_k +
(\mbox{div}_{S_r} K^k) (X) - D_X \theta_k
- \Eing(k,X)
\label{IdenEinKX}
\end{align}
where $\Eing$ is the Einstein tensor of $(\M,g)$.
\end{Prop}

\begin{proof}
Since $\s_{\ell}(X) = \frac{1}{2} \la \nabla_X k, \ell \ra$, we compute
\begin{align}
k(\s_{\ell}(X))
&=\frac{1}{2}\langle\nabla_k\nabla_{X}k,\ell\rangle
+\frac{1}{2}\langle\nabla_{X}k,\nabla_k\ell\rangle \nonumber \\
& = \frac{1}{2} \la \nabla_X \nabla_k k + \Riemg(k,X)k, \ell \ra
+ \frac{1}{2}\langle\nabla_{X}k,\nabla_k\ell\rangle \nonumber \\
& = - X(Q_k) + Q_k \s_{\ell}(X) + \frac{1}{2} \Riemg(\ell,k,k,X) 
+ \frac{1}{2}\langle\nabla_{X}k,\nabla_k\ell\rangle \label{difsa}
\end{align}
where in the second equality we used the definition of the curvature
operator $\Riemg$ and $[k,X]=0$ and in the third one we used
$\nabla_k k = Q_k k$. For the last term we use the immediate decomposition
\begin{align*}
\nabla_k \ell = - 2 \s^{\flat}_{\ell} - Q_k \ell
\end{align*}
where $\s^{\flat}_{\ell}$ is the vector metrically related to $\s_{\ell}$. Hence
$
\frac{1}{2}\langle\nabla_{X}k,\nabla_k\ell\rangle =
- K^k(X,\s^{\flat}_{\ell}) - Q_k \s_{\ell}(X)$
and (\ref{difsa}) becomes
\begin{align}
k(\s_{\ell}(X)) =
- X(Q_k) + \frac{1}{2} \Riemg(\ell,k,k,X) 
- K^k(X,\s^{\flat}_{\ell}). \label{inter1}
\end{align}
To ellaborate this further we use the Codazzi identity applied to $S_r$ along 
$k$ \cite{ONeill}
\begin{align*}
\Riemg(Y,X,Z,k) = (D_Y K^k)(X,Z) - (D_X K^k)(Y,Z) + \s_{\ell}(Y) K^k(X,Z)
- \s_{\ell}(X) K^k(Y,Z) 
\end{align*}
where $X,Y,Z \in \X(S_r)$ and $D$ is the covariant derivative of $S_r$.
Taking trace on $S_r$ in the first and third  indices and using that
\begin{align*}
\tr_{S_r} \left ( \Riemg(\cdot,X,\cdot,k) \right )
 = \Ricg(X,k) + \frac{1}{2} \Riemg(k,X,\ell,k)
= \Eing(X,k) + \frac{1}{2} \Riemg(k,X,\ell,k),
\end{align*}
which follows from the fact that $-\frac{1}{2} ( k \otimes \ell 
+ \ell \otimes k)$ is the metric of $(T_p S_r)^{\perp}$, we obtain
the contracted Codazzi identity along $k$ 
\begin{align}
\Eing(X,k) + \frac{1}{2} \Riemg(\ell,k,k,X) 
& = 
 (\mbox{div}_{S_r} K^k) (X) -  D_X \theta_k + K^{k}(\s^{\flat}_{\ell},X)
- \s_{\ell}(X) \theta_k \label{inter2} 
\end{align}
Eliminating $\Riemg(\ell,k,k,X)$ in (\ref{inter1})-(\ref{inter2})
yields the result.
\end{proof}

\section{Null Asymptotic flatness of $\Omega$
and asymptotic behaviour}
\label{Sect2}

The previous section involved general properties of $\Omega$ of local nature
and valid in any spacetime dimension. 
We now impose global conditions and restrict to dimension four.
First of all we
assume that $\Omega$ admits a global cross section $S_0$
(i.e. a smooth embedded spacelike
surface intersected precisely once by every inextendible curve along the 
null generators) of spherical topology $S_0$.  We also assume
that for one (and hence any) choice of geodesic null generator $k$ 
(i.e. satisfying $\nabla_k k =0$) the corresponding integral curve
starting at $p \in S_0$
has maximal domain $(-\infty,\lambda_{+}(p))$, i.e. the null
generators are past complete. 
After possibly removing portions of $\Omega$ lying to the future of $S_0$
we can assume
that $\Omega$ is foliated by the  level sets $\{ S_r \}$
of the function $r \in \F(\Omega)$ defined
by $k(r) = -1$, $r|_{S_0} =r_0$ and that all these level sets are diffeomorphic
to $S_0$ (so that in particular $\Omega = S_0 \times (R,\infty)$).
The function $r$ is called {\bf level set function of $k$}.
If we change the selection of null geodesic generator $k$,
the set of points to be removed is 
different, but since we are only interested in the past of $S_0$ this is irrelevant,
and we keep the same name $\Omega$.
A null hypersurface $\Omega$ satisfying these properties is called
{\bf extending to past null infinity}.

In order to define asymptotic flatness we need to impose decay of
various objects at infinity. First note that
covariant tensor fields $T$ on $\Omega$ completely
orthogonal to $k$ (i.e. satisfying
$T(k,\cdots) = T(\dots,k)=0$) are in one-to-one correspondance
with  smooth collections of covariant tensor fields $T(r)$ on each level
set $S_r$. We call such tensors {\bf transversal}. An immediate
example is the first fundamental form $\gamma$. The collection
of second fundamental forms $K^k_{S_r}$ defines a transversal
tensor denoted simply by $K^k$ (this is compatible with the
notation already used in the previous section). Similarly we have
the transversal tensors $K^{\ell}$ and $\s_{\ell}$. A transveral
tensor field $T$ is {\bf positive definite} iff $T(r)$ is positive
definite for all $r$. A transversal tensor field $T$ is called {\bf
Lie constant} iff $\pounds_{k} T =0$.

A vector field $X \in \X(S_0)$ can be extended uniquely to 
$\X(\Omega)$ by $[k,X]=0$. A local basis $\{ X_A \}$ on $S_0$
extended this way defines a
basis of each level set $S_r$. By $\{X_A\}$ we always mean any
such basis. A transversal tensor field $T$ on $\Omega$ 
is $T = O(1)$ iff $T_{A_1 \cdots A_q}:=
T(X_{A_1},\cdots,X_{A_q})$ is uniformly bounded. 
We write $T = O_n (r^{-q})$, $q \in \mathbb{R}$,
$n \in \mathbb{N}$ iff
\begin{align*}
r^q T = O(1), \quad r^{q+1} \pounds_k T = O(1), \quad \cdots \quad, 
r^{q+n} \underbrace{\pounds_{k} \cdots \pounds_k}_{n} T = O(1).
\end{align*}
We also write $T=o(r^{-q})$ iff $\lim_{r\rightarrow \infty}
r^q T(r)_{A_1\cdots A_q} =0$
and $T = o_n(r^{-q})$ iff $r^{i+q}
(\pounds_k)^i T = o(1)$ for all $i=0,1,\cdots,n$.
Given a transversal tensor (field) $T$ the tensor $\pounds_{X_A} T$
is also transversal. We write $ T = o^{X}_{n}(r^{-q})$ iff
\begin{align*}
r^q \underbrace{\pounds_{X_{A_1}} \cdots \pounds_{X_{A_i}}}_i T = o(1) \quad
\quad \forall i=0,1,\cdots, n
\end{align*}
It is clear that all these definitions are independent of the
choice of $\{ X_A \}$.

\begin{Def}
\label{AF}
Let $(\mathcal{M},g)$ be a four-dimensional spacetime.
A null hypersurface $\Omega$ is 
{\bf past asymptotically flat} if it 
extends to past null infinity and
there exists a choice of cross section
$S_0$ and null geodesic
generator $k$ with corresponding level set function $r$
with the following properties:
\begin{itemize}
\item[(i)]  There exist two symmetric 2-covariant transversal
and Lie constant  tensor fields $\qh$ and $h$ 
such that
$\gammatilde := \gamma - (r-r_0)^2 \qh - (r-r_0) h$
is $\gammatilde = o_1(r) \cap o^{X}_{2}(r)$ 
\item[(ii)] There exists a transversal, Lie constant
one-form $\sone$ such that
$\stilde_{\ell} := \s_{\ell} - \frac{\sone}{r-r_0} $
is $\stilde_{\ell} = o_1(r^{-1})$.
\item [(iii)] There exist Lie constant functions  $\thlo$ and
$\thlone$ such that
$\thltilde := \theta_{\ell} - \frac{\thlo}{r-r_0} - \frac{\A}{(r-r_0)^2}$ 
is $\thltilde = o(r^{-2})$.
\item[(iv)]
The scalar $\Riemg(X_A,X_B,X_C,X_D)$
along $\Omega$ is such that 
$\lim_{r\rightarrow \infty} \frac{1}{r^2}\Riemg(X_A,X_B,X_C,X_D)$ exists while
its double trace satisfies
$2\Eing(k,\ell)- \, \Scal^g - \frac{1}{2} \Riemg(\ell,k,\ell,k) = o (r^{-2})$. 
\end{itemize}
\end{Def}

This definition of asymptotic flatness is weaker than most existing
definition in the literature. Sometimes it will be convenient
to supplement it with a stronger notion where additional decay for 
some components of the Einstein tensor and for the remainder tensor $\tilde{\gamma}$ is assumed. 
Specifically, we say that an asymptotically flat null hypersurface $\Omega$ 
satisfies  the {\bf energy flux decay condition}
if 
\begin{align*}
\Eing(\ell,X_A)|_{\Omega} = o(r^{-2}), \quad \quad \pounds_{k} \gammatilde = o^{X}_{1}(1).  
\end{align*}

The name is motivated by the
analogous role of density flux that the Einstein tensor component
$\Eing(\ell,X_A)$ plays in the constraint equations for 
null hypersurfaces (see e.g. \cite{Mars2013}). 

Given a past asymptotically flat null hypersurface $\Omega$ with a choice
of $k$ and level set function $r$, it is convenient to define
$\rb := r - r_0$. The following Proposition determines
the asymptotic expansion of $K^k$ and provides an explicit
expression for $\thlo$.
\begin{Prop}
\label{leadingthetal}
Let $\Omega$ be a past asymptotically flat null hypersurface
with a choice of affinely parametrized null generator $k$ and
corresponding level set function $r$. Let $\gamma(r)^{AB}$
be the inverse of $\gamma(r)$. Then
\begin{align}
\gamma(r)^{AB} & = \frac{1}{\rb^2}\qh^{AB}-\frac{1}{\rb^3}\hh^{AB}+o(r^{-3}),
\label{metriccontravariant} \\
K^k_{AB} & =-\qh_{AB}\rb-\frac{1}{2}h_{AB}+ o(1), 
\label{Kk} \\
\theta_{\ell} & = \frac{2 \K_{\qh}}{\rb} + \frac{\A}{\rb^2} + o (r^{-2})
\label{thlo}
\end{align}
$\qh^{AB}$ is the inverse of $\qh_{AB}$, indices in
hatted tensors are raised and lowered with these
metrics and $\K_{\qh}$ is the Gauss curvature of $\qh_{AB}$. 
\end{Prop}

\begin{proof}
Expression (\ref{metriccontravariant}) is an immediate consequence
of item (i) in the definition of asymptotic flatness, namely
\begin{equation}
\label{metriccovariant}
\gamma_{AB}=\qh_{AB}\rb^2+h_{AB}\rb+o_1(\rb) \cap o^{X}_{2}(\rb). 
\end{equation}
For expression (\ref{Kk}) we use the standard identity
\begin{equation*}
k(\gamma_{AB})=k\langle X_A,X_B\rangle=\langle\nabla_k X_A,X_B\rangle+\langle X_A,\nabla_k X_B\rangle=2K^k_{AB}, 
\end{equation*}
the last step following from $[k,X_A]=0$. Inserting (\ref{metriccovariant})
yields (\ref{Kk}) immediately. Concerning the expansion for $\theta_{\ell}$
we invoke the Gauss identity for $(S_{r},\gamma(r))$, namely \cite{ONeill}
\begin{align*}
\Riemg(X_A,X_B,X_C,X_D) = \Riem^{\gamma(r)}_{ABCD} + \la K_{BC}, K_{AD} \ra
- \la K_{BD}, K_{AC} \ra
\end{align*}
where $K$ is the second fundamental form vector. Decomposing
$K = - \frac{1}{2} ( K^k \ell + K^{\ell} k)$ we have
\begin{align}
\Riem^g(X_A,X_B,X_C,X_D) = \Riem^{\gamma(r)}_{ABCD}-\frac{1}{2}(K^k_{BC} K^{\ell}_{AD} + K^k_{AD} K^{\ell}_{BC})
+\frac{1}{2}(K^k_{AC} K^{\ell}_{BD} + K^k_{BD} K^{\ell}_{AC}) 
\label{Gauss} 
\end{align}
The decomposition (\ref{metriccovariant}) implies that
$\Riem^{\gamma(r)}_{ABCD}= \rb^2 \Riem^{\qh}_{ABCD} + O(\rb)$ and given that 
$K^k_{AB} = - \qh_{AB} \rb + O(1)$, it follows from item (iv) in 
Definition \ref{AF} that $K^{\ell}_{AB}$ is of the form 
\begin{align}
K^{\ell}_{AB} = \rb K^{\ell}_{(0)}{}_{AB} + o(\rb)
\label{Kell2}
\end{align}
with $K^{\ell}_{(0)}$ a transverse 
Lie constant symmetric tensor on $\Omega$.
Taking trace in the $AC$
and $BD$ indices in (\ref{Gauss}) and using the fact that
$g^{\flat} = \gamma(r)^{AB} X_A \otimes X_B - \frac{1}{2} ( k \otimes \ell
+ \ell \otimes k)$ yields
\begin{align*}
2\Eing(k,\ell)- \, \Scal^g - \frac{1}{2} \Riemg(\ell,k,\ell,k) = 2 \K_{\gamma(r)} 
+ \theta_{\ell} \theta_k - K^{k}_{AB} K^{\ell}{}^{AB}.
\end{align*}
Now, $\K_{\gamma(r)} = \frac{1}{\rb^2} \K_{\qh} + o (r^{-2})$ and
the decompositions (\ref{metriccontravariant})-(\ref{Kk}), 
(\ref{Kell2})  and the
trace condition in item (iv) of Definition \ref{AF} imply
\begin{align*}
0 = \frac{1}{\rb^2} \left ( 2 \K_{\qh} - \thlo \right ) + o( r^{-2})
\end{align*}
which, together with item (iii) in Definition \ref{AF} proves (\ref{thlo}).
\end{proof}

\begin{remark}
Note that the expansion for $K^k$ only depends on item (i) in the
definition of asymptotic flatness. The expression for $\thlo$
depends on items (i), (iii) and (iv).
\end{remark}

\vspace{3mm}

\begin{remark}
We can raise the index to the tensor $K^k(r)$ with the contravariant
metric $\gamma^{AB}$. Combining the asymptotic expansions
(\ref{metriccontravariant}) and (\ref{Kk}) yields
\begin{align}
K^k(r)^{A}_{\phantom{A}B} = - \frac{1}{\rb} \delta^A_{\phantom{A}B}
+ \frac{1}{2} \hh^A_{\phantom{A}B}  \frac{1}{\rb^2} + o(r^{-2}) 
\label{KkEnd}
\end{align}
and taking trace
\begin{align}
\theta_k = 
- \frac{2}{\rb} + \frac{1}{2} (\tr_{\qh} \hh) \frac{1}{\rb^2}  +  
o(r^{-2}) 
:=  - \frac{2}{\rb} + \c \frac{1}{\rb^2}  + 
o(r^{-2}).
\label{thkexp}
\end{align}
\end{remark}

\vspace{3mm}

It will be convenient to endow each levet set
level set $\{S_r\}$  with a covariant derivative independent of $r$.
The natural choice is $\qh$ which is Lie constant, a metric on each $S_r$
and gives the leading term of the asymptotic expansion
of $\gamma(r)$. Denote by $\Dh$ the covariant derivate of $\qh$.
In the following lemma we find 
asymptotic expansion of the difference tensor
\begin{equation*} 
D_{X} Y-\Dh_{X} Y =Q(X,Y) 
\end{equation*}
and apply it to relate the Laplacians in the metrics $\gamma(r)$ 
and $\qh$  of a function. This will be needed
later when relating two different foliations on $\Omega$.
\begin{Lem}
The difference tensor $Q$ admits the decomposition
\begin{align}
Q^C_{AB} & =\frac{1}{2}(\Dh_A\hh^C_{\phantom{C}B}+\Dh_B\hh^C_{\phantom{C}A}-\Dh^C h_{AB})\frac{1}{\rb}+O(\rb^{-2}).
\label{ExpQ}
\end{align}
Moreover, if $F$ is a Lie constant function on $\Omega$ then
\begin{align}
\label{laplaciandev}
\triangle_\gamma \Ft &
=\triangle_{\qh} \Ft\frac{1}{\r^2}+\left(-\hh^{AB}\Dh_A\Dh_B \Ft-(\Dh_A\hh^{CA})
\Ft_{,C}+(\Dh^C\c)\Ft_{,C} \right)\frac{1}{\r^3}+o(\r^{-3})
\end{align}
\end{Lem}

\begin{proof}
We use the general formula for the difference tensor of Levi-Civita covariant
derivatives, see e.g. \cite{Wald},
\begin{equation*}
Q^C_{AB}=\frac{1}{2}\gamma^{CD}(\Dh_A\gamma_{DB}+\Dh_B\gamma_{DA}-\Dh_D\gamma_{AB}).
\end{equation*} 
Given that $\Dh_A\gamma_{DB}= (\Dh_A h_{DB}) \rb+O(1)$ and
$\gamma^{CD} = \frac{1}{\rb^2}\qh^{CD}+O(\rb^{-3})$, expression (\ref{ExpQ})
follows. For the Laplacian we use
\begin{align*}
\triangle_\gamma \Ft & = \gamma^{AB} D_A D_B \Ft =
\gamma^{AB} \left ( \Dh_A \Dh_B \Ft  - Q^C_{AB} \Dh_C \Ft \right ) \\
 & =\left(\frac{1}{\r^2}\qh^{AB}-\frac{1}{\r^3}\hh^{AB}+o(\r^{-3})\right)
\left(\Dh_{A}\Dh_{B}\Ft-\frac{1}{2}(\Dh_A\hh^C_{\phantom{C}B}
+\Dh_B\hh^C_{\phantom{C}A}-\Dh^C h_{AB})
\Ft_{,C}\frac{1}{\rb}+O(\rb^{-2})     \right)  \\
& =\frac{1}{\rb^2}\triangle_{\qh}\Ft+\left(-\hh^{AB}\Dh_{A}\Dh_{B}\Ft
-(\Dh_A\hh^{CA}) \Ft_{,C}
+ \frac{1}{2}\Dh^C(\qh^{AB}h_{AB})\Ft_{,C}  \right)\frac{1}{\rb^3}+o(\rb^{-3})
\end{align*}
which is (\ref{laplaciandev}) after recalling that $\tr_{\qh} \hh = 2 \c$.
\end{proof}

We conclude this section by showing that
the leading term $\sone$ of $\s_{\ell}$ is
fully determined in terms of the rest of objects
whenever the energy flux decay condition 
is assumed.
\begin{Prop}
\label{transversalProp}
Let $\Omega$ be a past asymptotically flat null hypersurface and
assume that the energy flux decay condition holds. Then
\begin{align*}
\sone_A =\Dh_A \c -\frac{1}{2}\Dh_B\hh^B_{\phantom{B}A}.
\end{align*}
\end{Prop}

\begin{proof}
From item (ii) in Definition \ref{AF} and the decomposition (\ref{thkexp})
we have 
\begin{equation}
\label{derivatives}
\pounds_{k} s_A=\frac{\sone_A}{\r^2}+o(r^{-2}), \quad \quad 
\theta_k \s_{\ell}{}_A
=- \frac{2\sone_A}{\r^2}+o(r^{-2})
\quad \quad  \mbox{and}
\quad \quad
\Dh_A \theta_k=
\frac{\Dh_A \c}{\r^2}+o(r^{-2})
\end{equation}
Concerning the divergence of $K^k{}^A_B$ in the metric $\gamma(r)$
we use that the leading term in (\ref{KkEnd}) is covariantly constant 
and then replace the $D-$covariant derivate
by the $\Dh$-derivative and use 
$Q^A_{AB} = O(\rb^{-1})$ to obtain
\begin{equation}
\label{covariantsigma}
D_B K^k{}^B_{\phantom{B}C}= 
\frac{1}{2 \rb^2} D_B \hh^{B}_{\phantom{B}C} + o(\rb^{-2}) = 
\frac{1}{2 \r^2} \Dh_B\hh^B_{\phantom{B}C} +o(\rb^{-2}).
\end{equation}
Thus, identity (\ref{IdenEinKX}) (with $Q_k=0$)  becomes
\begin{align*}
\frac{1}{\r^2} \left ( \sone_A+\frac{1}{2}\Dh_B\hh^B_{\phantom{B}A}-\Dh_A \c \right ) + o(\rb^{-2}) =0
\end{align*}
and the result follows.
\end{proof}

\section{Background foliation approaching large spheres}
\label{Sect3}

As discussed in the introduction, 
the Hawking energy has the interesting and well-known 
property  of approaching the Bondi energy when the surfaces
approach large spheres in a suitable sense. Our general
limiting expressions for the Hawking energy
will of course have to recover this fact.
To that aim, it is useful to restrict the choice of affinely
parametrized null generator  $k$ and corresponding level
set function $r$ so that the geometry of 
the level sets $S_r$ approaches, after rescalling,
the standard metric of unit radius on the sphere, denoted
by $\q$.

\begin{Def}
Let $\Omega$ be null and past asymptotically flat with a choice
of affine null generator $k$ and level set function $r$.
The foliation $\{S_r\}$ is said to {\bf approach 
large spheres} iff the leading term $\qh$ in the
expansion (\ref{metriccovariant}) of $\gamma$ is
the standard metric of a unit two-sphere.
\end{Def}

Our definition of approaching large spheres is equivalent to
demanding that the rescalled metric
$\frac{1}{(r-r_0)^2} \gamma(r)$, has a limit
$\q$ when $r\rightarrow \infty$. In \cite{Sauter2008} the definition
of approaching large spheres is defined more generally for exhaustions
$\{S_s \}$ of $\Omega$ with all elements diffeomorphic to each
other by demanding that the rescalled metric
$\frac{4\pi}{|S_s|} \gamma_{s}$ has a limit when 
$s \rightarrow \infty$ and defines a metric $\q$ of constant unit curvature.
It is clear that both definitions agree for the geodesic
foliations $\{S_r\}$ that we are using in this paper.

Our aim is to consider very general exhaustions $\{S_{r'}\}$ on
$\Omega$ and obtain the limit of the Hawking energy along them by
refering all objects to an affine background foliation
approaching large spheres. It is important to 
note that, since all Riemannian metrics on a manifold $\simeq
\mathbb{S}^2$ are conformal to the standard metric $\q$, 
there always exists
a (non-unique) choice of affine null generator $k$ in
an asymptotically flat
$\Omega$ with corresponding background foliation $\{S_r\}$ approaching
large spheres (cf. Remark \ref{RemMetric} below).
Tensors raised and lowered with the metric $\q_{AB}$ and its
inverse $\q^{AB}$
will have a circle on top, so that for instance 
(\ref{metriccontravariantsph}) reads
\begin{align}
\gamma(r)^{AB} & = \frac{1}{\rb^2}\q^{AB}-\frac{1}{\rb^3}\h^{AB}+o(r^{-3}).
\label{metriccontravariantsph} 
\end{align}
Note also that since $\q$ has constant unit curvature, the
null expansion $\theta_{\ell}$ has asymptotic behaviour
\begin{align*}
\label{thetaell}
\theta_{\ell} =
\frac{2}{\r}+\frac{\a}{\r^2}+o(\r^{-2}).
\end{align*}
We will consider three types
of foliations $\{S_{r'}\}$ and then combine them to obtain
the general case treated in Theorem \ref{main}.  Given a null basis
$\{k',\ell'\}$ orthogonal to a section $S$ in $\Omega$ and satisfying
$\la k',\ell'\ra = -2$, the mean curvature $H$ of $S$ decomposes as
$H = - \frac{1}{2} ( \theta_{k'} \ell' + \theta_{\ell'} k')$ and
the Hawking energy is
\begin{equation}
\label{hawkingmassexpansion}
m_H(S)=\sqrt{\frac{|S|}{16\pi}} \left ( 1+\frac{1}{16\pi}
\int_{S}\theta_{k'}\theta_{\ell'}\eta_{S} \right )
\end{equation}
so our aim will be to compute the limit of the
areas $|S_{r'}|$ and of 
$\theta_{k'}\theta_{\ell'}\eta_{S_{r'}}$ of the new foliation 
$\{ S_{r'}\}$ in terms of the background foliation geometry.
The following section deals with the case when
$\{S_{r'}\}$ is any geodesic foliation (not necessarily approaching large
spheres).

\section{Limit of the Hawking energy for geodesic foliations}
\label{Sect4}

In this section we assume that $\Omega$ is null
and past asymptotically flat and endowed with a foliation $\{S_r\}$
associated to an affinely parametrized null generator $k$
and approaching large spheres. By definition, 
a geodesic foliation 
in $\Omega$ is a foliation $\{ S_{r'} \}$
by cross sections with definiting function $r' \in \F (\Omega)$
such that the (unique) null generator $k'$ satisfying
$k'(r') = -1$ is affinely parametrized. Thus, there
exists a positive function $\phi \in \F(\Omega)$, Lie constant along 
$k$ such that $k' = \phi k$. Consequently, the level set functions
$r$ and $r'$
are necessarily related by
$r = \tau + \phi (r' - r_0)$ where $\tau \in \F(\Omega)$
is a Lie constant function. 
We first consider the case when
the two foliations $\{S_{r}\}$ and $\{S_{r'}\}$ have the same
starting surface, i.e. that $S_{r'=r_0} = S_{r=r_0}$, which 
fixes $\tau = r_0$. Thus, each of the surfaces $\{S_{r'}\}$ can be described
by a graph function $r= r_0 + \phi (r' - r_0)$. 

Our strategy is to use the general expressions in Section \ref{Setup} 
for the geometry
of a graph $r=F$ in a background foliation, insert
the resulting expressions in (\ref{hawkingmassexpansion}) and take the
limit  when $r' \rightarrow \infty$. We start with the following
Proposition.

\begin{Prop}
\label{geodesichange}
Let $\Omega$ be a past asymptotically flat null hypersurface endowed
with an affinely parametrized background foliation
$\{S_r\}$ with generator $k$ and approaching large spheres. Let
$\phi>0$ be a Lie constant function and define the geodesic foliation
$\{S_{r'}\}$ by the graph functions 
\begin{equation*}
r = F_{r'} := r_0+\phi \rt, \quad \quad \rt := r' - r_0.
\end{equation*}
Let $k'=\phi k$ and $\ell'$ be the null normal orthogonal to $\{S_{r'}\}$
satisfying $\la k',\ell' \ra = -2$. Then the 
induced metric $\gamma(r')$, volume form 
$\bm{\eta_{S_{r'}}}$, null expansions $\theta_{k'}$, $\theta_{\ell'}$ and the connection of the normal bundle ${s_{\ell'}}$ 
can be expressed in terms of the background geometry as   
\begin{align}
\gamma'_{AB}& :=\gamma'(X'_A,X'_B)  = \phi^2 \q_{AB}\rt^2+\phi h_{AB}\rt
+o(\rt)  \label{gammaprime} \\
\bm{\eta_{S_{r'}}}
= & \left ( \phi^2\rt^2+ \phi \c \rt+o(\rt) 
\right ) \volunitdos \label{volprime} \\
\theta_{k'}  = &
- \frac{2}{\rt}+\frac{\c}{\phi}\frac{1}{\rt^2}+o(\rt^{-2}) 
\label{thetakhatexpansion}
 \\
\theta_{\ell'} = &
-\frac{2}{\phi^2}(\triangle_\q\log\phi-1)\frac{1}{\rt} +
\left(
\frac{\a}{\phi^3}-\frac{4\h^{AB}\phi_{,A}\phi_{,B}}{\phi^5} 
+\frac{\c|\esf\phi|^2_\q}{\phi^5}+\frac{4\q^{AB}\phi_{,A}(s_B)_1}{\phi^4}
\right . \nonumber \\
& \left . 
+\frac{2\h^{AB}\esf_A\esf_B\phi}{\phi^4}+\frac{2\esf_A\h^{AC}\phi_{,C}}{\phi^4}
-\frac{2\esf^C\c\phi_{,C}}{\phi^4} \right)\frac{1}{\rt^2} +o(\rt^{-2}) 
 \label{thetalhatexpansion}
\\
{s_{\ell'}}_A= &
=\left(\frac{\sone_A}{\phi}-\frac{\phi_{,L}\h^L_{\phantom{L}A}}{2\phi^2}\right)\frac{1}{\rt}+o(\rt^{-1}) 
 \label{slhatexpansion}
\end{align}
where $X'_A := X_A - X_A (F_{r'}) k$. 
\end{Prop}

\begin{proof}

As already mentioned $k'$ is the generator of the foliation $\{S_{r'}\}$. For any point
$p \in S_{r'}$, the level set passing through $p$ has
$\rb = F_{r'}(p) -r_0 = \phi(p) \rt$. Thus (\ref{metric}) 
and the background expansion (\ref{metriccovariant}) with
$\qh \rightarrow \q$ gives
\begin{equation*}
\gamma'_{AB} |_p = \gamma_{AB} |_{S_{r=r_0 + \phi(p) \rt}} =
\q_{AB} \rb^2 + h_{AB} \rb + \gammatilde |_{\rb = \phi \rt}
\end{equation*}
which is (\ref{gammaprime}). (\ref{volprime}) follows by taking determinants and
using the standard identity
\begin{equation*}
\det (M + s B) = (\det M) ( 1 + s \, \tr \, (M^{-1} B) + O(s^2) ),
\end{equation*}
valid for any invertible matrix $M$. For 
$\theta_{k'}$ we use the fact that 
$\theta_k$ is a property of 
$\Omega$ and not of the surface embedded  
in $\Omega$ passing through that point, cf. Corollary \ref{RelExpansions}.
Thus
\begin{equation}
\theta_{k'}  = \phi \theta_k|_{\rb=\phi \rt} =
-\frac{2}{\rt}+\frac{\c}{\phi}\frac{1}{\rt^2}+o(\rt^{-2}),
\end{equation}
as claimed in (\ref{thetakhatexpansion}).
To compute $\theta_{\ell'}$ we use expression (\ref{Relthl}) 
with $Q_{k}=0$, $\alpha = \phi$ and graph function 
$F = F_{r'} = r_0 + \phi \rt$. As before, the right-hand side
has to be evaluated at $\rb = \phi \rt$. We work out each term separately:
\begin{align*}
\theta_\ell & =\frac{2}{\r}+\frac{\a}{\r^2}+o(\r^{-2})
=\frac{2}{\phi}\frac{1}{\rt}+\frac{\a}{\phi^2}\frac{1}{\rt^2}+o(\rt^{-2}), \\
\theta_k |DF_{r'}|^2 & = \theta_k \gamma^{AB} \nabla_A F_{r'}
\nabla_B F_{r'} \\
& = \left( - \frac{2}{\phi}\frac{1}{\rt}+\frac{\c}{\phi^2}\frac{1}{\rt^2}
+o(\rt^{-2}) \right)
\left( \frac{\q^{AB}}{\phi^2}\frac{1}{\rt^2}
-\frac{\h^{AB}}{\phi^3}\frac{1}{\rt^3}+O(\rt^{-4})
\right)(\phi_{,A}\rt)(\phi_{,B}\rt) \nonumber \\
&=- \frac{2|\esf\phi|^2_\q}{\phi^3}\frac{1}{\rt}
+\left(\frac{2\h^{AB}\phi_{,A}\phi_{,B}}{\phi^4}+\frac{\c|\esf\phi|^2_\q}{\phi^4}
\right)\frac{1}{\rt^2}+o(\rt^{-2}).
\end{align*}
For the Laplacian of $F_{r'}$ we use (\ref{laplaciandev}) which 
gives, using $\nabla F_{r'} = (\nabla \phi) \rt$,
$\Hess_{\q} F_{r'} = (\Hess_{\q} \phi ) \rt$ and
$\triangle_{\q} F_{r'} = (\triangle_{\q} \phi ) \rt$,
\begin{equation*}
\triangle_{\gamma} F_{r'}=\frac{\triangle_\q\phi}{\phi^2}\frac{1}{\rt}
+\left(-\frac{\h^{AB}\esf_A\esf_B \phi}{\phi^3}
-\frac{\esf_A\h^{CA}\phi_{,C}}{\phi^3}
+\frac{(\esf^C\c)\phi_{,C}}{\phi^3} \right)\frac{1}{\rt^2}+o(\rt^{-2}).
\end{equation*}
For the term $\s_{\ell} (\grad F_{r'})$
\begin{align*}
\s_{\ell} (\grad F_{r'}) & = \gamma^{AB} 
\s_{\ell}{}_A \nabla_B F_{r'} 
 =
\left(\frac{\q^{AB}}{\phi^2}\frac{1}{\rt^2}+o(\rt^{-2})\right)
\left(\frac{\sone{}_A}{\phi}\frac{1}{\rt}+o(\rt^{-1})\right) \phi_{,B} \rt \\
& =  \frac{\q^{AB} \sone_{A} \phi_{,B}}{\phi^3}\frac{1}{\rt^2}+o(\rt^{-2})
\end{align*}
and finally $K^k(\grad F_{r'}, \grad{F}_{r'})$ is, inserting
(\ref{metriccontravariant}) and (\ref{Kk}),
\begin{align*}
K^k(\grad F_{r'}, \grad{F}_{r'})   = &
\gamma^{AB} \gamma^{CD} K^k_{BD} \phi_{,A} \phi_{,B}\rt^2 = \frac{1}{\phi^4 \rt^2} 
\left ( \q^{AB}-\frac{1}{\phi \rt}\h^{AB}+o (\rt^{-1} )
\right ) \times \\
& \left ( \q^{CD}-\frac{1}{\phi \rt}\h^{CD}+o (\rt^{-1} )
\right )
\left ( -\q_{BD} \phi \rt -\frac{1}{2}h_{BD}+o(1) \right )
\phi_{,A} \phi_{,B} \\
 = &  \frac{-|\esf\phi|^2_\q}{\phi^3}\frac{1}{\rt}+\frac{3\h^{LA}\phi_{,A}\phi_{,L}}{2\phi^4}\frac{1}{\rt^2}+o(\rt^{-2}).
\end{align*}
Putting things together leads to (\ref{thetalhatexpansion})
after using $\frac{\triangle_{\q} \phi}{\phi} - \frac{|\esf \phi|^2_\q}{\phi^2} =
\triangle_{\q} \log \phi$.

Finally,
if we substitute in expression (\ref{oneformalpha}) the asymptotic 
expansion of item (ii) in Definition \ref{AF}, and the 
expansion of $K^k$ of (\ref{KkEnd}), we have  
\begin{eqnarray*}
{s_{\ell'}}_A=-\frac{\phi_{,A}}{\phi}+\frac{\sone_A}{\phi}\frac{1}{\rt}+\frac{\phi_{,A}}{\phi}-\frac{\phi_{,L}\h^L_{\phantom{L}A}}{2\phi^2}\frac{1}{\rt}+o(\rt^{-1})=\left(\frac{\sone_A}{\phi}-\frac{\phi_{,L}\h^L_{\phantom{L}A}}{2\phi^2}\right)\frac{1}{\rt}+o(\rt^{-1}),
\end{eqnarray*}
where the term in $\frac{1}{\rt}$ is $\sonep_A$.

\end{proof}

\begin{remark}
\label{RemMetric}
The foliation $\{ S_{r'} \}$ has limit metric $\lim_{r' \rightarrow
\infty} \frac{\gamma'}{\rt^2} = \phi^2 \q := \qh$. Using the formula for the 
scalar curvature of a conformal metric it follows
\begin{align}
\label{Gausscurv}
\K_{\qh} = \frac{1 - \triangle_{\q} \log \phi}{\phi^2}
\end{align}
so that the expansion of $\theta_{\ell'}$ is $\theta_{\ell'} =
\frac{2 \K_{\qh}}{\rt} + o(\rt^{-1})$ in agreement with Proposition 
\ref{leadingthetal}. 
Note that the transformation law $\qh \rightarrow \phi^2 \qh'$ 
for the leading term 
$\qh$ in the metric $\gamma(r)$ under change of foliation $k' = \phi k$
holds irrespectively of whether the background foliation approaches
large spheres or not. Since as mentioned above, any
metric on $\mathbb{S}^2$ is conformal to $\q$, it follows
that any asymptotically flat $\Omega$ admits a background foliation
approaching large spheres.

\end{remark}

\vspace{3mm}

We can now evaluate the limit of the
Hawking energy  for the foliation $\{S_{r'}\}$:
\begin{Tma}
\label{TmaHawkinglimit}
Let $\Omega$ be a past asymptotically flat null hypersurface endowed
with an affinely parametrized background foliation
$\{S_r\}$ with generator $k$ and approaching large spheres. Let
$\Psi>0$ be a Lie constant function and define the geodesic foliation
$\{S_{r'}\}$ by the graph functions 
\begin{equation*}
r = F_{r'} := r_0+ \frac{1}{\Psi} \rt, \quad \quad \rt := r' - r_0
\end{equation*}
The limit of the Hawking energy along the foliation $\{S_{r'} \}$
is, in terms of the background geometry,
\begin{equation}
\label{hawkingmasslimit}
\underset{r' \to \infty}{\lim}
m_H(S_{r'})=\frac{1}{8\pi\sqrt{16\pi}}
\left(\sqrt{\int_{\mathbb{S}^2}\frac{1}{\Psi^2}\volunitdos}\right)
\int_{\mathbb{S}^2}\left(\triangle_\q\c-(\c+\a)-
4 \mbox{div}_{\q} (\sone) \right)\Psi\volunitdos.
\end{equation}
\end{Tma}

\begin{proof}
Set $\phi = \Psi^{-1}$ so that we can use
the expressions in  Proposition \ref{geodesichange}.
We need to compute $\theta_{k'} \theta_{\ell'} \bm{\eta_{S_{r'}}}$. 
Denoting 
by $\B$ the coefficient of the term $\frac{1}{\rt^2}$ in $\theta_{\ell'}$
we immediately find, from (\ref{volprime})-(\ref{thetalhatexpansion}),
\begin{align*}
\theta_{k'}\theta_{\ell'}\bm{\eta_{S_{r'}}} 
&= \left ( 4(\triangle_\q \log\phi-1)
+(- 2 \K_{\qh} \c\phi-2\B \phi^2)\frac{1}{\rt}+o(\rt^{-1}) \right ) \volunitdos
\end{align*}
and hence
\begin{align*}
1 + \frac{1}{16\pi} \int_{S_{r'}} \theta_{k'}\theta_{\ell'} 
\bm{\eta_{S_{r'}}} & = 
1 + \frac{1}{16 \pi} \int_{\mathbb{S}^2} 
\left ( 4(\triangle_\q \log\phi-1)
+(- 2 \K_{\qh} \c\phi-2 \B \phi^2)\frac{1}{\rt}+o(\rt^{-1}) \right ) \volunitdos \\
& = \frac{1}{16 \pi \rt} 
\int_{\mathbb{S}^2} 
\underbrace{(- 2 \K_{\qh} \c\phi-2\B \phi^2)}_{\I} \volunitdos
+o(\rt^{-1}).
\end{align*}
Concerning the area term
\begin{equation*}
|S_{r'}|=\int_{S_{r'}}\bm{\eta_{S_{r'}}}=
\left(\int_{\mathbb{S}^2}\phi^2\volunitdos\right)\rt^2+O(\rt)
\Longrightarrow
\sqrt{|S_{r'}|}= \rt \sqrt{\int_{\mathbb{S}^2}\phi^2 \volunitdos} +O(1)
\end{equation*}
so that the limit of the Hawking energy is 
\begin{align}
\underset{r' \to \infty}{\lim}m_H(S_{r'}) & 
=\frac{1}{16\pi\sqrt{16\pi}}\left(\sqrt{\int_{\mathbb{S}^2}\phi^2\volunitdos}
\right)\int_{\mathbb{S}^2} \I \volunitdos.
\label{Hmasslimit}
\end{align}
We now compute $\I$. Using the explicit form of $\K_{\qh}$ and $\B$ it follows
\begin{align*}
\I  = & 
\frac{2\c\triangle_{\q}\phi}{\phi^2}-\frac{4\c|\esf\phi|^2_{\q}}{\phi^3}
-\frac{2}{\phi}(\c+\a)
+\frac{8\h^{AB}\phi_{,A}\phi_{,B}}{\phi^3}
-\frac{8\q^{AB}\phi_{,A} \sone{}_B }{\phi^2} \\
& - \frac{4\h^{AB}\esf_A\esf_B\phi}{\phi^2}
-\frac{4\esf_A\h^{AB}\phi_{,B}}{\phi^2}
+\frac{4\esf^B\c\phi_{,B}}{\phi^2} \\
 = & 
\esf_A \left ( - \frac{4}{\phi^2} \h^{AB} \phi_{,B} 
+ \frac{2 \c \esf^A \phi}{\phi^2}
+ \frac{8 \sone{}^A}{\phi} 
\right ) 
+\frac{2\esf^A\c\phi_{,A}}{\phi^2} 
-\frac{2}{\phi}(\c+\a)
- \frac{8}{\phi} \mbox{div}_{\q} (\sone) \\
 = & 
\esf_A \left ( - \frac{4}{\phi^2} \h^{AB} \phi_{,B} 
+ \frac{2 \c \esf^A \phi}{\phi^2}
+ \frac{8 \sone{}^A}{\phi} 
- \frac{2}{\phi} \esf^{A} \c \right )  
+ \frac{2}{\phi} \triangle_{\q} \c
-\frac{2}{\phi}(\c+\a)
- \frac{8}{\phi} \mbox{div}_{\q} (\sone) 
\end{align*}
and (\ref{Hmasslimit}) becomes (\ref{hawkingmasslimit})
after using the Gauss identity and $\phi = \Psi^{-1}$.
\end{proof}

\begin{remark}
It is interesting that all terms involving derivatives of $\phi$ (or
$\Psi$) combine themselves into a divergence and drop out after
intergration. The behaviour of the limit of the Hawking energy under
change of geodesic foliation is hence much simpler than one might
have expected a priori. 
\end{remark}

\vspace{3mm}

Given a past asymptotically  flat null hypersurface, there are many 
possible choices
of geodesic background foliations approaching large spheres. 
Any two such foliations are related by $\r = \phi \rt$ with 
$\phi$ satisfying 
\begin{equation}
\label{largesphereeq}
\triangle_\q\log\phi+\phi^2=1
\end{equation}
so that  the Gauss curvature (\ref{Gausscurv}) of $\qh$ is also one.
In this case the limit of the Hawking
energy of the foliation $\{S_{r'}\}$ can be computed in two different
ways, namely refering $\{ S_{r'} \}$ to the background folition $\{S_r\}$
and using Theorem \ref{TmaHawkinglimit} or considering $\{S_{r'}\}$ itself
as a background folition (so that the result would be (\ref{hawkingmasslimit})
with $\Psi=1$ and $\c$, $\a$, $\sone$ all referred to the
foliation $\{ S_{r'}\}$). It is clear that both results must agree. This
requires a kind of covariance property of the integral in 
(\ref{hawkingmasslimit}). Remarkably, this covariance occurs already
at the level of the integrand, as we show next.
All geometric objects referred to the geodesic foliation $\{S_{r'}\}$
will carry a prime.

\begin{Tma}
\label{Tmalargespheres}
Let $\Omega$ be a past asymptotically flat null hypersurface endowed
with an affinely parametrized background foliation
$\{S_r\}$ with generator $k$ and approaching large spheres.
Let
$\phi>0$ be a Lie constant function and define the geodesic foliation
$\{S_{r'}\}$ by the graph functions 
\begin{equation}
\label{parameterchangetwo}
r = F_{r'} := r_0+ \phi \rt, \quad \quad \rt := r' - r_0,
\end{equation}
with $\phi$ satisfying (\ref{largesphereeq}). Then
\begin{equation*}
\left(\triangle_\q\c-(\c+\a)
- 4 \mbox{div}_{\q} \, \sone
\right) \frac{1}{\phi} \volunitdos=\left( \triangle_{\q'}\c'-(\c'+\a')
- 4 \mbox{div}_{\q'} \, \sonep \right) 
\bm{\eta_{\q'}}, 
\end{equation*} 
As a consequence we have the necessary invariance of the limit 
of the Hawking energy
\begin{equation*}
\underset{r' \to \infty}{\lim}
m_H(S_{r'})=\frac{-1}{16\pi}\int_{\mathbb{S}^2}\left( \c'+\a' \right)\bm{\eta_{\q'}}
=\frac{1}{16\pi}
\int_{\mathbb{S}^2}\left(\triangle_\q\c-(\c+\a)-
4 \mbox{div}_{\q} \, \sone \right) \frac{1}{\phi} \volunitdos.
\end{equation*}
\end{Tma}

\begin{proof}
The general expressions (\ref{gammaprime})-(\ref{thetalhatexpansion}) give the
explicit form 
of the geometric objects of $\{ S_{r'} \}$ in terms
of the background folition $\{ S_r \}$, namely
\begin{align*}
\q'_{AB} & =\phi^2 \q_{AB}, \quad \quad \c'=\frac{\c}{\phi},  
\quad \quad 
\a'  = \frac{\a}{\phi^3}-\frac{4\h^{AB}\phi_{,A}\phi_{,B}}{\phi^5} 
+\frac{\c|\esf\phi|^2_\q}{\phi^5}+ 
 \frac{4\q^{AB}\phi_{,A}(s_B)_1}{\phi^4} \\
& +\frac{2\h^{AB}\esf_A\esf_B\phi}{\phi^4}+\frac{2\esf_A\h^{AC}\phi_{,C}}{\phi^4}
-\frac{2\esf^C\c\phi_{,C}}{\phi^4},
\quad \quad
\sonep_A 
= \frac{\sone_A}{\phi}-\frac{\phi_{,L}\h^L_{\phantom{L}A}}{2\phi^2}.
\end{align*}
The Laplacian in two-dimensions is conformally covariant so that
$\triangle_{\q'}\c'
=\frac{1}{\phi^2}\triangle_\q\c'$, and hence
\begin{align*}
\triangle_{\q'}\c' =
\frac{1}{\phi^2}\triangle_\q\c'=\frac{\triangle_\q\c}{\phi^3}
-\frac{\c\triangle_\q\phi}{\phi^4}-\frac{2\esf^A\c\phi_{,C}}{\phi^4}
+\frac{2\c|\esf\phi|^2_\q}{\phi^5}.
\end{align*}
The divergence of a one-form in two-dimensions is also 
conformally covariant 
$\mbox{div}_{\q'} \sonep = 
\frac{1}{\phi^2} \mbox{div}_{\q} \sonep $ and
\begin{align*}
\mbox{div}_{\q'} \sonep 
=\frac{\esf^A\sone_A}{\phi^3}-\frac{\sone_A\esf^A\phi}{\phi^4}+\frac{1}{\phi^5}\phi_{,A}\phi_{,L}\h^{AL}-\frac{1}{2\phi^4}\h^{AL}\esf_A\esf_L\phi-\frac{1}{2\phi^4}\phi_{L}\esf_A\h^{LA}.
\end{align*}
Putting things together many terms cancell out and we find
\begin{align*}
\triangle_{\q'}\c'-(\c'+\a')-4 \mbox{div}_{\q'} \sonep & 
= \frac{\triangle_\q\c}{\phi^3}-
\frac{\c}{\phi^3} \left ( \frac{\triangle_\q \phi}{\phi} -
\frac{|\esf\phi|^2_\q}{\phi^2}
\right )
-\frac{\c}{\phi}-\frac{\a}{\phi^3}- \frac{4}{\phi^2}  \mbox{div}_{\q} \sone
\end{align*}
Using the large sphere equation
(\ref{largesphereeq})
$\frac{\triangle_{\q} \phi}{\phi} - \frac{|\esf \phi|^2_\q}{\phi^2} =
\triangle_{\q} \log \phi = 1 - \phi^2$ we find
\begin{equation*}
\triangle_{\q'}\c'-(\c'+\a')-
4 \mbox{div}_{\q'} \sonep 
=\left(\triangle_\s\c-(\c+\a)-4 \mbox{div}_{\q} \sone \right ) 
\frac{1}{\phi^3}.
\end{equation*}
Since the volume forms
are related by $\bm{\eta_{\q'}}=\phi^2\volunitdos$, 
the result follows.
\end{proof}

We have considered so far the limit of the Hawking energy for geodesic
foliations with fixed initial surface $S_0$. The second step is to 
consider geodesic foliations with a different initial surface. Let us fix
a geodesic background foliation $\{ S_{r} \}$
aproaching large spheres and, as usual, let $k$
be the associated null generator $k$ satisfying $k(r)=-1$. Any geodesic
foliation with the same null generator is defined by the equation $r' = 
\mbox{const}$, where $r'$ is any solution of $k(r')=-1$. Hence 
$k(r-r') =0$ and the function $\tau := r - r'$ is Lie constant. This function
can be interpreted as the graph function of the initial
surface $S_{r'=r_0}$ in the original foliation $\{S_{r'} \}$. If $\{S_{r'}\}$
starts at a larger initial value $r'_0$, the graph function of
$S_{r'_0}$ is given by an appropiate constant shift of $\tau$, namely
$\tau + r'_0 + r_0$. The following theorem gives the limit of the Hawking energy 
for the foliation $\{S_{r'} \}$ and shows that the integrand is also 
covariant (in fact invariant) for this change of foliation.

\begin{Tma}
\label{Tmatau}
Let $\Omega$ be a past asymptotically flat null hypersurface endowed
with an affinely parametrized background foliation
$\{S_r\}$ with generator $k$ and approaching large spheres.
Let
$\tau$ be a Lie constant function and define the geodesic foliation
$\{S_{r'}\}$ by the graph functions 
\begin{equation}
r = F_{r'} := \tau+r',
\end{equation}
Then we have $\q' = \q$,
\begin{equation*}
\theta_k=\frac{-2}{\rbp}+\frac{\c+2\tau}{\rbp^2}+o(\rbp^{-2}), \quad  \theta_{\ell'}=\frac{2}{\rbp}+\frac{\a-2\tau-2\triangle_\q\tau}{\rbp^2}+o(\rbp^{-2}), \quad  {s_{\ell'}}_A=\frac{\sone_A+\tau_{,A}}{\rbp}+o(\rbp^{-1}), 
\end{equation*}
and
\begin{equation}
\label{bracketcovariance}
\left(\triangle_\q\c-(\c+\a)-4 \esf^A\sone_A\right)\volunitdos=
(\triangle_{\q'}\c'-(\c'+\a')-4\nabla^{\q' A}(s'_A)_1)\bm{\eta_{\q'}}.
\end{equation} 
Moreover, the
limits of the Hawking energies along the two foliations
coincide and
read
\begin{equation}
\underset{r \to \infty}{\lim}
m_H(S_{r})=
\underset{r' \to \infty}{\lim}
m_H(S_{r'})=\frac{-1}{16\pi}
\int_{\mathbb{S}^2}(\c+\a)\volunitdos=\frac{-1}{16\pi}
\int_{\mathbb{S}^2}(\c'+\a')\bm{\eta_{\q'}}.
\end{equation}

\end{Tma}

\begin{proof}
We write, as before, $\r = r - r_0$ and $\rt = r'-r_0$, so that
$\r = \rt + \tau$. Changing parameters in
the first fundamental form
\begin{equation*}
\gamma_{AB}=\q_{AB}\rb^2+h_{AB}\rb+\Psi_{AB}=\q_{AB}(\tau+\rbp)^2+O(\rbp)=
\q_{AB}\rbp^2+O(\rbp)
\end{equation*}
it follows  $\q'=\q$. Similarly, (\ref{thkexp}) gives
\begin{equation*}
\theta_k=\frac{-2}{\rb}+\frac{\c}{\rb^2}+o(\rb^{-2})=\frac{-2}{\tau+\rbp}+\frac{\c}{(\tau+\rbp)^2}+o(\rbp^{-2})=\frac{-2}{\rbp}+\frac{\c+2\tau}{\rbp^2}+o(\rbp^{-2})
\end{equation*}
so that  $\c'=\c+2\tau$. For $\theta_{\ell'}$ we 
use (\ref{Relthl}) with $\alpha=1$, $Q_{k}=0$ and $F = F_{r'}$. 
First we change  parameter in the expansion of the first terms
in the right-hand side
\begin{equation*}
\theta_\ell=\frac{2}{\rb}+\frac{\a}{\rb^2}+o(\rb^{-2})=\frac{2}{\tau+\rbp}+\frac{\a}{(\tau+\rbp)^2}+o(\rbp^{-2})=\frac{2}{\rbp}+\frac{\a-2\tau}{\rbp^2}
+o(\rbp^{-2}).
\end{equation*}
Since 
$D F_{r'} = D \tau$ independent of $r$ and $\gamma^{AB} = O(\rt^{-2})$,
cf. (\ref{metriccontravariant}), we have 
$|DF_{r'}|^2_\gamma = O (\rt^{-2})$ and 
the term $\theta_k |DF_{r'}|^2_{\gamma} = O(\rt^{-3})$. The same
argument shows that 
$\s_{\ell} (\grad F) = O(\rt^{-3})$
and  $K^{k} ( \grad F, \grad F) = O(\rt^{-3})$.  
The Laplacian term can be computed, using (\ref{laplaciandev}), as
\begin{equation*}
\triangle_\gamma F_{r'}
=\frac{\triangle_\q\tau}{\rb^2}+o(\rb^{-2})
=\frac{\triangle_\q\tau}{\rbp^2}+o(\rbp^{-2}).
\end{equation*}
Inserting also these into expression 
(\ref{Relthl}) yields
\begin{equation*}
\theta_{\ell'}=\frac{2}{\rbp}+\frac{\a-2\tau}{\rbp^2}
-\frac{2\triangle_\q\tau}{\rbp^2}+o(\rbp^{-2})
=\frac{2}{\rbp}+\frac{\a-2\tau-2\triangle_\q\tau}{\rbp^2}+o(\rbp^{-2}),
\end{equation*}
as claimed. Note in particular that $\a'=\a-2\tau-2\triangle_\q\tau$.
The expansion of the connection one-form $\sonep$ is obtained 
from (\ref{oneformalpha}) with $\alpha=1$ and $Q_k=0$. The term
$-K^{k}(X_A,\mbox{grad} \, F_{r'})$ is
$-K^{k}(X_A,\mbox{grad}\, F_{r'}) = -K^k(X_A,\mbox{grad}\, \tau) =
 \frac{1}{\r} \tau_{A} + O(\r^{-2})$ after using (\ref{KkEnd}). Given 
the expansion of $s_{\ell}$ in item (ii) of 
Definition \ref{AF} we conclude
\begin{equation*}
{s_{\ell'}}_A=\frac{\sone_A}{\rb}+\frac{\tau_{,A}}{\rb}+o(\frac{1}{\rb})
=\frac{\sone_A+\tau_{,A}}{\rb}+o(\frac{1}{\rb})=
\frac{\sone_A+\tau_{,A}}{\rbp}+o(\frac{1}{\rbp}) 
\quad \Longrightarrow \quad
\sonep_A=\sone_A+\tau_{,A}. 
\end{equation*}
Inserting $\q'$, $\c'$, $\a'$ and $\sone'$ in 
$(\triangle_{\q'}\c'-(\c'+\a')-4\nabla^{\q' A}(s'_A)_1)\bm{\eta_{\q'}}$ 
all terms in $\tau$ cancel out and the invariance 
(\ref{bracketcovariance}) is established.
For the last statement we use the fact
that both foliations $\{S_r\}$ and $\{ S_{r'} \}$ are
geodesic foliations approaching large spheres. So, both can be taken 
as background foliations and in each case we can 
apply Theorem \ref{TmaHawkinglimit} 
with $\Psi =1$. Using the Gauss identity,
the equalities
\begin{align*}
\underset{r \to \infty}{\lim} m_H(S_{r})
=\frac{-1}{16\pi}
\int_{\mathbb{S}^2}(\c+\a)\volunitdos \quad \mbox{and}
\quad 
\underset{r' \to \infty}{\lim} m_H(S_{r'})
=\frac{-1}{16\pi}\int_{\mathbb{S}^2}(\c'+\a')\bm{\eta_{\q'}}
\end{align*}
hold. Invoking the invariance (\ref{bracketcovariance}) 
the remaining equality 
$\underset{r \to \infty}{\lim} m_H(S_{r}) =
\underset{r' \to \infty}{\lim} m_H(S_{r'})$ follows.
\end{proof}

\section{Limit of the Hawking energy for non-geodesic foliations}
\label{Sect5}

Our aim in this paper is to obtain the limit of the Hawking energy
along very general foliations $\{S_{r^{\star}} \}$. In the
previous section we dealt with the general case when the foliations are
geodesic. In order to go into more general settings we need to
consider vector fields $k'$ not affinely parametrized. We will
however assume that $k'$ is nowhere zero even at the limit at infinity. More
specifically, we assume that there exists a function
$r' \in \mathcal{F}(\Omega)$ satisfying $k'(r')=-1$ and a geodesic
background foliation (not necessarily approaching large spheres) 
defined as the level sets of a function $r \in \mathcal{F}(\Omega)$ 
such that $\fun:=r -r'$ decays at infinity in a appropriate way. In other words,
the foliation $\{S_{r'}\}$ is assumed to approach at infinity
a geodesic foliation $\{S_{r}\}$ at an appropriate rate. Conversely, given 
a geodesic background foliation and
a function $\fun \in \mathcal{F} (\Omega)$ satisfying
$\fun = o_1(1)$ we can define a function $r' := r -\fun \in \mathcal{F} 
(\Omega)$.
The level sets of this function are smooth surfaces at least for
points at large enough $r$. This is because $dr' (k) = k(r -\fun) = 
-1 - k(\fun) \neq 0$ because $k(\fun)$ decays at infinity. Thus, for $r$
bigger than some (possible large) value $R_1$, the level sets $r' = \mbox{const}$
define a foliation $\{S_{r'}\}$. Each surface on this foliation is transverse
to $k$ and hence spacelike. The null generator $k'$ satisfying $k'(r')= -1$
is given by $k' = \frac{1}{1 + k(\fun)} k$ because
\begin{align*}
k'(r') = \frac{1}{1 + k(\fun)} k(r') = \frac{1}{1 + k(\fun)} 
k(r - \fun) = -1.
\end{align*}
It is clear that the foliation $\{ S_{r'} \}$ is not geodesic in general.
Given a value $r'$ large enough, the 
surface $S_{r'}$ is a graph on the brackground foliation $\{ S_{r} \}$.
The graph function $r=F_{r'}$ is given by $F_{r'}(p) = r' + \fun(p)$ for all
$p \in S_{r'}$. As usual, we extend the graph function to $\Omega$ 
by Lie dragging along  $k$. Note that $F_{r'}$ extended this way is {\it not} $r' + \fun$, but both agree on $S_{r'}$. Thus we can safely abuse notation an
write the graph simply as $F_{r'} = r' + \fun$.
The following theorem gives the limit of the
Hawking energy for the foliation $\{S_{r'}\}$.

\begin{Tma}
\label{Tmaq}
Let $\Omega$ be a past asymptotically flat null hypersurface endowed
with an affinely parametrized background foliation
$\{S_r\}$ with generator $k$. Let
$\fun \in \mathcal{F}(\Omega)$ satisfy 
$\fun=\ordone\cap\ordtwo$, and $k(\fun)=\ordthree$
. Define the foliation
$\{S_{r'}\}$ by the graph functions 
\begin{equation}
\label{parameterchange}
r = F_{r'} :=r'+\fun.
\end{equation}
Then the null expansions and the connection one-form 
of $\{S_{r'}\}$ have, to leading orders, the same form as 
for the background foliation, i.e.
\begin{equation}
\theta_{k'}=\frac{-2}{\rb'}+\frac{\c}{\rb'^2}+o(\rb'^{-2}), \quad
\theta_{\ell'}=\frac{2\K_{\hat{q}}}{\rb'}+\frac{\a}{\rb'^2}+o(\rb'^{-2}) 
\quad {s_{\ell' A}}=\frac{(s^{(1)}_{\ell})_A}{\r'}+o(\r'^{-1}),
\label{case3} 
\end{equation}
where $\rb := \r - r_0$,
$\rb' := \rb'-r_0$ and all objects refer to the background foliation. The limit of the Hawking energy along $\{S_{r'}\}$
is the same as the limit along $\{S_r\}$ and reads
\begin{equation}
\underset{r \to \infty}{\lim}m_H(S_{r})=\underset{r' \to \infty}{\lim}m_H(S_{r'})=\frac{-1}{8\pi\sqrt{16\pi}}\sqrt{|\hat{S}|}\int_{\hat{S}}\left(\K_{\hat{q}}\c+\a \right)\bm{\eta_{\qh}}, \label{Limit2}
\end{equation}
where $|\hat{S}|$ is the area of $(S_0,\qh)$ 
and $\bm{\eta_{\hat{q}}}$ the corresponding volume form.   
\end{Tma}

\begin{proof}
The asymptotic expansion for $\theta_{k}$ is obtained by simply
changing the parameter of the foliation $\rb = \rb' + \fun$
\begin{eqnarray*}
\theta_{k}=\frac{-2}{\rb}+\frac{\c}{\rb^2}+o(\rb^{-2})=\frac{-2}{\rb'+\fun}+\frac{\c}{(\rb'+\fun)^2}+o(\rb'^{-2})=\frac{-2}{\rb'}+\frac{\c}{\rb'^2}+o(\rb'^{-2})
\end{eqnarray*}
where $\fun=\ordone$  has been used in the last equality.  Given that 
$k'=\left( \frac{1}{1+k(\fun)}\right)k$ and
$k(\fun) = o(\rb^{-1})$, the null expansion along $k'$ is
\begin{eqnarray*}
\theta_{k'}=\left( \frac{1}{1+k(\fun)}\right)\theta_k=\frac{-2}{\rb'}+\frac{\c}{\rb'^2}+o(\rb'^{-2}),
\end{eqnarray*}
which is the first expression in (\ref{case3}).
We next compute $\theta_{\ell'}$ from (\ref{Relthl}) with 
$\alpha=\frac{1}{1+k(\fun)}$. Changing parameters in the first term
of the right hand side,
\begin{equation*}
\theta_{\ell}=\frac{2 \K_{\qh}}{\rb'+\fun}+\frac{\a}{(\rb'+\fun)^2}+o(\rb'^{-2})=\frac{2\K_{\qh}}{\rb'}+\frac{\a-2\fun}{\rb'^2}+o(\rb'^{-2})=\frac{2\K_{\qh}}{\rb'}+\frac{\a}{\rb'^2}+o(\rb'^{-2}).
\end{equation*}
For the terms involving the graph function, it is immediate to check that
$|DF_{r'}|^2_{\gamma}\theta_k= o(\rb'^{-3})$,
$\s_{\ell} (\grad F_{r'}) = o(\rb'^{-3})$, $K^{k} ( \grad F_{r'}, \grad F_{r'}) 
= o(\rb'^{-3})$.
For the Laplacian term (\ref{laplaciandev}) gives
\begin{equation*}
\triangle_{\gamma} F_{r'}=\frac{\triangle_{\qh}F_{r'}}{\rb^2}+o(\rb^{-2})=\frac{\triangle_{\qh}F_{r'}}{\rb'{^2}}+o(\rb'^{-2})=\frac{\triangle_{\qh}\fun}{\rb'{^2}}+o(\rb'^{-2}) = o(\rb'^{-2}), 
\end{equation*}
because $\fun= \ordtwo$. Finally, $k(\fun) = o(\rb'^{-1})$ implies
$\frac{1}{\alpha}=1+k(\fun)=1+ o(\rb'^{-1}) $ and (\ref{Relthl}) is simply
\begin{equation*}
\theta_{\ell'}=( 1+o(\rb'^{-1}) )\left( \frac{2\K_{\qh}}{\rb'}+\frac{\a}{\rb'^2}+o(\rb'^{-2}) \right)=\frac{2\K_{\qh}}{\rb'}+\frac{\a}{\rb'^2}+o(\rb'^{-2})
\end{equation*} 
as stated in the Theorem. The connection one-form ${s_{\ell'}}$
is obtained from (\ref{oneformalpha}) with $Q_k=0$ and
$\alpha= \frac{1}{1+k(\fun)}$. Given that
$\alpha_{,A} = o(\rb'^{-1})$ (because $k(\fun) = o_1^X(\rb'^{-1})$) and
$s_{\ell} (\grad F_{r'}) = o(\rb'^{-2})$, we conclude
\begin{equation*}
{s_{\ell' A}}=\frac{(s^{(1)}_{\ell})_A}{\r'}+o(\r'^{-1}).
\end{equation*}
We next compute the limit of the Hawking energy along $\{S_{r'}\}$. 
The metric in $S_{r'}$ is 
\begin{equation*}
\gamma(r')_{AB}=\hat{q}_{AB}\rb^2+o(\rb^2)=\hat{q}_{AB}(\rb'+\fun)^2+o(\rb'^{2})=\hat{q}_{AB}\rb'^2+o(\rb'^{2}),
\end{equation*}
so that in particular the rescaled limit metric
and corresponding volume forms
remain unchanged, $\qh'=\qh$ and $\bm{\eta_{\qh'}}=\bm{\eta_{\qh}}$.
This, together with the expansions (\ref{case3}), already implies that 
the limit of the Hawking energy along $\{ S_{r'} \}$ and along $\{ S_{r} \}$
are the same. To obtain expression (\ref{Limit2}) we need
the volume form of $S_{r'}$.  As with the metric $\gamma(r)$
or with the null expansion $\theta_{k}$ it suffices to change parameter
in the volume form $\bm{\eta_{S_r}}$ which is given by
(\ref{volprime}) with $\phi=1$
\begin{equation}
\label{volumetransf}
\bm{\eta_{S_{r'}}}= \left( (\rb'+\fun)^2+\c(\rb'+\fun)+o(\rb') \right)
\bm{\eta_{\qh}}=(\rb'^2+\c \rb'+o(\rb'))\bm{\eta_{\qh}}.
\end{equation}
It is immediate to check that the product 
$\theta_{k'}\theta_{\ell'}\bm{\eta_{S_{r'}}}$ is
\begin{eqnarray*}
\theta_{k'}\theta_{\ell'}\bm{\eta_{S_{r'}}}
=\left( -4\K_{\hat{q}}+(-2\K_{\hat{q}}\c-2\a)\frac{1}{\rb'}+o(\rb'^{-1})  \right)\bm{\eta_{\qh}}
\end{eqnarray*}
so that, using Gauss-Bonnet
$\int_{\hat{S}}\K_{\hat{q}}\bm{\eta_{\qh}}=4\pi$, 
\begin{align*}
1+\frac{1}{16\pi}\int_{S_{r'}}\theta_{k'}\theta_{\ell'}\bm{\eta_{S_{r'}}}
& =\frac{1}{16\pi}\int_{\hat{S}}(-2\K_{\hat{q}}\ch-2\ah)\bm{\eta_{\qh}}
\frac{1}{\rb'}+o(\rb'^{-1}).
\end{align*}
On the other hand $|S_{r'}|=\int_{\hat{S}}(\rb'^2+\ch\rb'+o(\rb'))
\bm{\eta_{\qh}}=|\hat{S}|\rb'^2+o(\rb'^2)$$\Rightarrow$
 $\sqrt{|S_{r'}|}=\sqrt{|\hat{S}|}\rb'+o(\rb')$ and
\begin{align*}
m_H(S_{r'}) & =\frac{1}{\sqrt{16\pi}}\left(\sqrt{|\hat{S}|}\rb'+o(\rb') \right)\left(\frac{1}{16\pi}\int_{\hat{S}}(-2\K_{\hat{\s}}\ch-2\ah)\bm{\eta_{\qh}}\frac{1}{\rb'}+o(\rb'^{-1})\right) \nonumber \\
&=\frac{-1}{8\pi\sqrt{16\pi}}\sqrt{|\hat{S}|}\int_{\hat{S}}(\K_{\hat{\s}}\ch+\ah)\bm{\eta_{\qh}}+o(1).
\end{align*}

\end{proof}

We are ready to obtain our main Theorem \ref{main} by simply combining
the previous results.  In fact, we state and prove a slightly 
more complete theorem
that provides two different expressions for the limit.
\begin{Tma}[\bf General Hawking energy limit]
\label{tmaarbitrary}
Let $\Omega$ be a past asymptotically flat null hypersurface endowed
with an affinely parametrized background foliation
$\{S_r\}$ with generator $k$ that tends to large spheres.
Define the foliation $\{S_{r^*}\}$ by the graph functions
  \begin{equation}
\label{generalchangeparameter}
r=F_{r'}:=r_0+ \frac{1}{\Psi} (r^{\star}-r_0)+\tau+\fun,
\end{equation}
with $\Psi>0$, $\tau$ Lie constant functions on $\Omega$ and 
$\fun= \ordone\cap \ordtwo$ and $k(\fun)= \ordthree$.  The limit
of the Hawking energy along $\{ S_{r^{\star}} \}$ is
\begin{align}
\underset{r^{\star} \to \infty}{\lim}m_H(S_{r^{\star}})&
=\frac{-1}{8\pi\sqrt{16\pi}}
\left(\sqrt{\int_{\mathbb{S}^2} \bm{\eta_{\hat{q}}}
} \right)
\int_{\mathbb{S}^2}\left(\K_{\hat{q}} \c{}^{\star}+\a{}^{\star} \right)
\bm{\eta_{\hat{q}}}
\nonumber \\
\label{hawkingmasslimittwo}
& =\frac{1}{8\pi\sqrt{16\pi}}\left(\sqrt{\int_{\mathbb{S}^2} \frac{1}{\Psi^2}
\volunitdos}\right)\int_{\mathbb{S}^2}\left(\triangle_\q\c-(\c+\a)-
4 \mbox{div}_{\q} (\sone) 
\right) \Psi \volunitdos,
\end{align}
where $\qh$, $\c{}^{\star}$ and $\a{}^{\star}$ refer either
to the foliation  $\{S_{r'} \}$ or to the geodesic foliation $\{ S_{r''} \}$
defined by the graph functions $r = F_{r''} := r_0 
+ \frac{1}{\Psi} (r -r_0) + \tau$, and $\q$, $\c$, $\a$ and $\sone$
refer to the background foliation $\{ S_{r} \}$.
\end{Tma}

\begin{proof}
The strategy is to pass from the background foliation to
$\{ S_{r^{\star}} \}$ in three steps. The geometric elements
of each foliation use the same symbol as the foliation, so the meaning
of each quantity should be clear.
Consider first a foliation defined by the
levet sets of $r' := r - \tau$.  Theorem \ref{Tmatau} gives
\begin{equation*}
\label{bracketeq}
\left(\triangle_\q\c-(\c+\a)-4
\mbox{div}_{\q} (\sone) 
\right)\volunitdos= \left (\triangle_{\q'}\c'-(\c'+\a')-4
\mbox{div}_{\q'} (\sone')  \right )\bm{\eta_{\q'}}.
\end{equation*} 
Consider next the foliation defined by the level sets
of $r''$, where $r'=r_0+ \Psi^{-1} (r''-r_0)$. Since $\{ S_{r'} \}$ is a geodesic
foliation approaching large spheres, Theorem 
\ref{TmaHawkinglimit} implies the limit of the Hawking energy
is 
\begin{equation*}
\underset{r'' \to \infty}{\lim}m_H(S_{r''})=\frac{1}{8\pi\sqrt{16\pi}}\left(\sqrt{\int_{\mathbb{S}^2}\frac{1}{\Psi^2} \bm{\eta_{\q'}}    }\right)\int_{\mathbb{S}^2}\left(\triangle_{\q'}\c'-(\c'+\a')-
4 \mbox{div}_{\q'} (\sone')  \right)\Psi\bm{\eta_{\q'}},
\end{equation*}
which, upon using (\ref{bracketeq}) and $\volunitdos=\bm{\eta_{\q'}}$, implies
\begin{equation}
\label{hawkingbefore}
\underset{r'' \to \infty}{\lim}m_H(S_{r''})=\frac{1}{8\pi\sqrt{16\pi}}\left(\sqrt{\int_{\mathbb{S}^2}\frac{1}{\Psi^2}\volunitdos}\right)\int_{\mathbb{S}^2}\left(\triangle_\q\c-(\c+\a)-
4 \mbox{div}_{\q} (\sone) 
\right)\Psi\volunitdos.
\end{equation}
Now, $\{S_{r''}\}$ is geodesic but does not necessarily tend
to large spheres.  The final foliation $\{ S_{r^{\star}} \}$ 
is related to $\{ S_{r''} \}$ by $r'' = r^{\star} + \Psi \fun$. Since
$\Psi \fun$ satisfies the hypotheses of Theorem   
\ref{Tmaq} we conclude that the Hawking energy has the
same limit
along $\{ S_{r''} \}$ and along $\{ S_{r^{\star}} \}$. In combination with
(\ref{hawkingbefore}) this proves the second equality in
(\ref{hawkingmasslimittwo}).  For the first equality we simply note that 
the rescaled limit metric of $\{ S_{r''} \}$ is $\hat{q} = \Psi^{-2} \q$
and apply again Theorem \ref{Tmaq}.
\end{proof}

\vspace{3mm}

\begin{remark}
\label{futurenull}
In this paper we have considered null hypersurfaces extending
to past null infinity. Obviously similar results apply for asymptotically
flat null hypersurfaces extending to future null infinity.  By repeating the
arguments before, the following result is obtained:
consider
a future directed geodesic null vector
$\bar{k}$ tangent to $\Omega$ and define the
function $r \in \mathcal{F} (\Omega)$ by $\bar{k}(r)=1$ with $r=r_0$
on some initial cross section. The level sets $\{S_{r}\}$ define a foliation
which allows to construct a transversal future directed null normal 
$\bar{\ell}$ satisfying $\la \bar{k}, \bar{\ell} \ra = -2$.
If the rescalled asymptotic metric of the foliation
is spherical $\q$, the expansions of the null second fundamental forms
and connection one-form take the form (note the change of signs with
respect to the past null case)
\begin{equation*}
\theta_{\bar{k}}=\frac{2}{\rb}+\frac{\theta_{\bar{k}}^{(1)}}{\rb^2}+o(\rb^{-2}), \quad \quad \theta_{\bar{\ell}}=\frac{-2}{\rb}+\frac{\theta_{\bar{\ell}}^{(1)}}{\rb^2}+o(\rb^{-2})  \quad \quad {s_{\bar{\ell}}}_A=\frac{ {s_{\bar{\ell}}}_A^{(1)}    }{\rb}+o(\rb^{-1})  
\end{equation*}
and the limit of the Hawking energy along $\{S_{r^*}\}$ (with the same
definition as in Theorem \ref{tmaarbitrary}) is
\begin{equation*}
\underset{r^* \to \infty}{\lim}
m_H(S_{r^*})=\frac{1}{8\pi\sqrt{16\pi}}
\left(\sqrt{\int_{\mathbb{S}^2} \frac{1}{\Psi^2} \volunitdos}\right)
\int_{\mathbb{S}^2}\left(-\triangle_\q\theta_{\bar{k}}^{(1)}+(\theta_{\bar{k}}^{(1)}+\theta_{\bar{\ell}}^{(1)})-
4 \mbox{div}_{\q} (s_{\bar{\ell}}^{(1)}) \right)\Psi\volunitdos.
\end{equation*}

\end{remark}

\vspace{3mm}

\section{The large sphere equation and the Bondi energy-momentum}
\label{Sect6}

As mentioned in the introduction, the limit of the Hawking energy
when the foliation approaches large spheres is the Bondi energy. In this
section we want to recover this fact from our general expressions. Recall
first that the conformal group of the two-sphere is defined
as the set of diffemorphisms $\Phi : (\mathbb{S}^2,\q)
\mapsto (\mathbb{S}^2,\q)$ satisfying $\Phi^{\star}(\q) = \Theta^2 \q$,
$\Theta  \in \F(\mathbb{S}^2,\mathbb{R}^+)$ (i.e.
the set of conformal diffeomorphisms).
We restrict ourselves to the connected component of the identity of this
group. It is well-known (see e.g. \cite{PenroseRindler})
that this group is isomorphic to
the connected component of the identity of Lorentz
group of Minkowski
space $\mathbb{M}^{1,3}$, and also isomorphic to the M\"obius group of the Riemann sphere
\begin{align}
F : \mathbb{S}^2 & \mapsto \mathbb{S}^2 \\
          z & \mapsto F(z) = \frac{\alpha z + \beta }{\gamma z + \delta}
          \quad \quad \left ( \begin{array}{cc}
                                \alpha & \beta \\
                                \gamma & \delta 
                               \end{array}
\right ) \in SL(2,\mathbb{C}) \label{Moebius}
\end{align}
where $z \in \mathbb{C} \, \cup \{ \infty \} \simeq \mathbb{S}^2$. In these coordinates,
the standard metric on the sphere is $\q = \frac{4}{(1 + z \zet)^2} dz
d \zet$ and the $l=1$ spherical harmonics read
\begin{equation}
\label{sphericalharmonics}
Y^1_{1}=\frac{z+\zet}{1+z\zet},\quad \quad 
Y^1_{2}=\frac{z-\zet}{i(1+z\zet)},\quad \quad Y^1_{3}=\frac{z\zet-1}{1+z\zet}.
\end{equation}
For a vector $a \in \mathbb{R}^3$ we write 
$a \cdot Y^1 \defi \sum_{i=1}^3 a^i Y^1_i$.
These properties allow us to obtain easily the general solution to the large
sphere equation (\ref{largesphereeq}). 

\begin{Prop}[\bf Solution of the large sphere equation]
A smooth  function $\phi: \mathbb{S}^2 \mapsto  \mathbb{R}^+$
solves equation (\ref{largesphereeq})
if and only if there exists $a=(a^{1},a^2,a^3) \in \mathbb{R}^3$ such that
\begin{equation}
\label{solutionlargespheres}
\Psi := \frac{1}{\phi} = \sqrt{1+|a|^2}+a\cdot Y^1.
\end{equation} 
\end{Prop}

\begin{proof}
In terms of $\Psi:=\frac{1}{\phi}$, equation (\ref{largesphereeq}) becomes
\begin{equation}
\label{largesphereeqtwo}
\Psi^2+(\triangle_\q\Psi)\Psi-|\esf\Psi|^2_\q=1.
\end{equation}
We first show that (\ref{solutionlargespheres}) solves
this equation. Applying a rotation to
$\mathbb{S}^2$ we can assume without loss of generality that 
$a = (0,0,c)$ and hence $\Psi=\sqrt{1+c^2}+c Y^1_3$. Thus
$\triangle_\q\Psi=c \triangle_\q Y^1_3=- 2 c Y^1_3 = -2 \Psi +
2 \sqrt{1+c^2}$, and $|\esf\Psi|^2_\q
=(1+z\zet)^2\partial_z\Psi\partial_{\zet}\Psi
= \frac{4c^2 z\zet}{(1+z\zet)^2} = - ( \Psi^2 + 1 - 2 \sqrt{1 +c^2} \, \Psi )$, and (\ref{largesphereeqtwo}) holds after immediate
cancellations.

To show the converse we recall that equation (\ref{largesphereeq}) is
the statement that the Gauss curvature
$\K_{\phi^2 \q} = 1$. This means that there exist coordinates
$z' \in \mathbb{C} \cup \{ \infty \}$ where $\phi^2 \q = \frac{4 dz'
d\zet^{\prime}}{(1 + z' \zet^{\prime})^2}$. We can assume without loss of
generality that the map $F(z) = z'$ is orientation preserving. Since it is also
an element of the conformal group, it must be an element
of the M\"obius group (\ref{Moebius}). Performing the pull-back of $\q$
\begin{align*}
\phi^2 \q = \frac{4 |\frac{\partial F}{\partial z}|^2}{\left ( 1 + |F|^2
\right )^2} dz d \zet.
\end{align*}
Thus $\phi = \frac{ 1 + |z|^2 }{1 + |F|^2} \left
| \frac{\partial F}{\partial z} \right |$. Since $\frac{\partial F}{\partial z}
= \frac{1}{(\gamma z+ \delta)^2}$ if follows
\begin{equation*}
\Psi = \frac{1}{\phi}=\frac{|\alpha z+\beta|^2+|\gamma z+\delta|^2}{(1+|z|^2)}.
\end{equation*}
Expanding in terms of $l=1$ spherical harmonics yields
\begin{equation*}
\Psi = \frac{|\alpha|^2 + |\beta|^2 + |\gamma|^2 + |\delta|^2}{2} 
+ \mbox{Re} \left ( \overline{\alpha} \beta + \overline{\gamma} \delta \right )
Y^1_{1} 
+ \mbox{Im} \left ( \overline{\alpha} \beta + \overline{\gamma} \delta \right )
Y^1_{2} + \frac{|\alpha|^2 - |\beta|^2 + |\gamma|^2 - |\delta|^2}{2} Y^1_{3}
\end{equation*}
It is straightforward to check that this expression is of the form
$\Psi = \sqrt{1 + |a|^2} + a \cdot Y^1$ with $a \in \mathbb{R}^3$.
\end{proof}

\begin{remark}
The Bondi energy-momentum is a vector in an abstract Minkowski space.
Let us recall the construction for the sake of completeness and because
of a subtlety that arises in the case of past null hypersurfaces.
The Lorentz transformation $x^{\prime}{}^{\mu} = \Lambda(F)^{\mu}_{\phantom{\mu}\nu} x^{\nu}$
associated to the
M\"obius transformation $F$ has a time direction $\partial_{t'}$ given by
\begin{align*}
u := \partial_{t'} = \Lambda(F)^{0}_{\phantom{0}0} \partial_t - \sum_{i=1}^3
\Lambda(F)^0_{\phantom{0}i} \partial_{x^{i}}.
\end{align*}
The explicit map $\Lambda(F)$ can be found e.g. in page 17 of
\cite{PenroseRindler}. Comparing $\Lambda^0_{\phantom{0}i}(F)$
with the expresion for
$a^i$ above it follows
\begin{align*}
u = \sqrt{1 + |a|^2} \partial_t - a^i  \partial_{x^i}.
\end{align*}
The construction of $\Lambda(F)$ in \cite{PenroseRindler} is performed
with the unit sphere lying at the intersection
of the hyperplane $t=1$ and the future null cone of the origin.
It is hence adapted to future directed
null hypersurfaces extending to future null infinity. In this paper we have
considered  null hypersurfaces extending to past null infinity. This case
is obtained from the previous one by a time inversion, which has the efect that
the observer $u$ has the form
\begin{align*}
u = \sqrt{1 + |a|^2} \partial_t + a^i  \partial_{x^i}
\end{align*}
in terms of the coefficients $a^i$ in the conformal factor $\phi$.
Summarizing, to a background foliation $\{S_r\}$ of 
$\Omega$  approaching large spheres
with asymptotic rescalled metric $\q$  one can asign 
an asymptotic inertial reference frame $\{ t, x^i \}$ in an (abstract)
Minkowski spacetime. Given another such foliation
$\{ S_{r'} \}$ with asymptotic rescalled metric $\phi^2 \q$, one associates
an asymptotic inertial observer with time direction
$u^{\mu} = (\sqrt{1 + |a|^2}, a^i)$ in the basis $\partial_{x^{\alpha}}$ above.
\end{remark}

\vspace{3mm}

We can now recover the result that the Hawking energy approaches the Bondi
energy for spherical foliations.

\begin{Cor}
\label{fourmomentum}
Let $\Omega$ be a past asymptotically flat null hypersurface endowed
with an affinely parametrized background foliation
$\{S_r\}$ with generator $k$ that tends to large spheres.
Consider another foliation associated to the parameter $r^*$ so that $r=r_0+\phi(r^*-r_0)+\tau+\fun$,  
as in Theorem \ref{tmaarbitrary}, where $\phi >0$
satisfies the large sphere equation (\ref{largesphereeq}).  Let
$u^{\mu} \in \mathbb{M}^{1,3}$ be the asymptotic inertial observer
associated to this foliation. Then
\begin{equation*}
\underset{r^* \to \infty}{\lim}m_H(S_{r^*})=-P_B^{\mu}u^{\nu}\eta_{\mu\nu} := E^u_B, 
\end{equation*}
where $\eta_{\mu\nu}$ is the Minkowski metric
and the Bondi four-momentum vector $P_B$ reads
\begin{align}
\label{p0}
E_B & :=P_B^{0}:=\frac{-1}{16\pi}\int_{\mathbb{S}^2}(\c+\a)\volunitdos \\
\label{pi}
P_B^{i} & :=\frac{1}{16\pi}\int_{\mathbb{S}^2}\left(- \triangle_\q\c + (\c+\a)+
4 \mbox{div}_{\q} \, \sone \right)Y^1_{i}\volunitdos, \quad i\in\{1,2,3\}.
\end{align}
If, in adition, the energy flux decay condition of Proposition
\ref{transversalProp} is satisfied, then the Bondi three-momentum simplifies to
\begin{equation}
P_B^{i}=\frac{1}{16\pi}\int_{\mathbb{S}^2}\left(\triangle_\q\c+(\c+\a)
\right)Y^1_{i}\volunitdos, \quad i\in\{1,2,3\}.
\label{Pbi}
\end{equation}
\end{Cor}

\begin{proof}
We can use expression (\ref{hawkingmasslimittwo})
with $\Psi$ as in
(\ref{solutionlargespheres}) so that
\begin{align}
\underset{r^* \to \infty}{\lim} & m_H(S_{r^*})  
=  \frac{1}{16\pi}\int_{\mathbb{S}^2}
\left(\triangle_\q\c-(\c+\a)-4 \mbox{div}_{q} \, \sone 
\right) \left (\sqrt{1+|a|^2}+\sum_{i=1}^3 a^i Y^1_i \right )\volunitdos  \label{IntMH} \\
= & \left ( \frac{-1}{16\pi}\int_{\mathbb{S}^2}(\c+\a)\volunitdos
\right)\sqrt{1+|a|^2}+  \nonumber \\
\label{hawkingmasslim}
&+\sum_{i=1}^3\left( \frac{1}{16\pi}\int_{\mathbb{S}^2}\left(\triangle_\q\c-(\c+\a)-4 \mbox{div}_{\q} \, \sone \right)Y^1_i\volunitdos \right)a^i 
=-\eta(u,P_B),
 \end{align}
with $u=(\sqrt{1+|a|^2},a^1,a^2,a^3)$ and $P^i_B$ as given in the
statement of the Corollary.

When the  energy flux decay condition holds, we have from Proposition \ref{transversalProp}
$\sone_A=\esf_A \c-\frac{1}{2}\esf_B\h^B_{\phantom{B}A}$, and
the integral (\ref{IntMH}) becomes
\begin{align*}
\int_{\mathbb{S}^2}\left(-3\triangle_\q\c-(\c+\a)+2\esf_A\esf_B \h^{AB}\right)\Psi_{a} \volunitdos
\end{align*}
where $\Psi_a := \sqrt{1 + |a|^2} + a \cdot Y^1$. Integrating by parts
the last term and using that $\mbox{Hess}_{\q} \Psi_a= - (a\cdot Y^1) \q
= \frac{1}{2} (\triangle_{\q} \Psi_{a}) \q$ yields
\begin{align*}
\int_{\mathbb{S}^2} 2 \esf_A\esf_B \h^{AB}\Psi\volunitdos
=  \int_{\mathbb{S}^2} 2 \c \triangle_{\q} \Psi_a \volunitdos
= \int_{\mathbb{S}^2} 2 (\triangle_{\q} \c)  \Psi_a \volunitdos
\end{align*}
where in the first equality we used $\tr_{\q} \h = 2 \c$ and in
the second we performed another integration by parts. Arguing as before, the 
expresion (\ref{Pbi}) for $P^i_B$ follows.
\end{proof}

\begin{remark}
An analogous result can be obtained for the case of asymptotically flat null
hypersurfaces $\Omega$ approaching future null infinity. Using the
general expression in Remark \ref{futurenull} for the limit of the Hawking energy 
in this case and using the fact that $\Psi_{a} = 
\sqrt{1 + |a|^2} + a \cdot Y$ corresponds now to the asymptotic
observer with four velocity $u^{\alpha} = ( \sqrt{1 + |a|^2}, - a^1, -a^2, -a^3)$,
the Bondi energy-momentum vector $P_B$ satisfying $\lim_{r^{\star}
\rightarrow \infty} M_H (S_{r^{\star}}) = - \eta(u,P_B) := E^{u}_B$ is
\begin{align*}
E_B & :=P_B^{0}:=
\frac{1}{16\pi}\int_{\mathbb{S}^2}(\c+\a)\volunitdos \\
P_B^{i} & :=
\frac{1}{16\pi}\int_{\mathbb{S}^2}\left(- \triangle_\q\c + (\c+\a)
- 4 \mbox{div}_{\q} \, \sone \right)Y^1_{i}\volunitdos.
\end{align*}
The energy flux decay condition in this case implies (i.e. 
the analogous on Proposition \ref{transversalProp})
\begin{align*}
\sone_A = - \Dh_A \c -\frac{1}{2}\Dh_B\hh^B_{\phantom{B}A}
\end{align*}
and the Bondi momentum simplifies to
\begin{align*}
P_B^{i}=\frac{1}{16\pi}\int_{\mathbb{S}^2}\left(\triangle_\q\c+(\c+\a)
\right)Y^1_{i}\volunitdos, \quad i\in\{1,2,3\}.
\end{align*}
Note that in this case
\begin{align}
E^u_B := - \eta_{\alpha\beta} u^{\alpha} P_B^{\beta} 
= \frac{1}{16\pi}\int_{\mathbb{S}^2}\left(\triangle_\q\c+(\c+\a)
\right) (u^0 - u^i Y^1_{i} )\volunitdos. \label{EuB}
\end{align}
\end{remark}

\vspace{3mm}

\begin{remark}
The relationship between the limit of the
Hawking energy and the Bondi four-momen\-tum for foliations approaching
large spheres has been investigated in \cite{PenroseRindler}
and \cite{Bartnik2004} 
(see also Definition 4.2 in \cite{Sauter2008}). As a useful
check, it is convenient to see how the results in this paper
fit with the results in \cite{Bartnik2004}. The setup there involves
so-called null quasi-spherical coordinates which are adapted
to a foliation by future outgoing null hypersurfaces $\{ {\mathcal N}_z \}$,
each of them
foliated by codimension-two spacelike surfaces $S_{z,r_{B}}$ (we change
Bartnik's notation $r$ to $r_{B}$ to avoid conflict with our notation
above). Each $S_{z,r_{B}}$  has induced metric
isometric to the standard sphere of radius $r_{B}$. In fact, the null
quasi-spherical coordinates $\{z,r_{B},\theta,\phi\}$ are such that
the surface $\{ z=\mbox{const}, r_{B} = \mbox{const}\}$ has induced metric
$r_{B}^2 (d \theta^2 + \sin^2 \theta d \phi) := r_B^2 \q$, which selects the 
diffeomorphism of $S_{r_{B},z}$ with the standard unit sphere
$(\mathbb{S}^2,\q)$. Under asymptotic
conditions along the null hypersurface involving the shear and its
angular derivative, Bartnik shows  among various other things
that the Bondi energy-momentum is well-defined and agrees with
the limit of the Hawking energy along the quasi-spherical foliation
$S_{z,r_{B}}$. More precisely, defining the mass aspect function 
$m = \frac{1}{2} r_{B} \left ( 1 - \frac{1}{4} \vec{H}^2 r_{B}^2 \right )$
of the sphere $S_{z,r_{B}}$  so that $m_H(z,r_{B})=\frac{1}{4\pi}\int_{S^2_{z,r_{B}}}m 
\volunitdos$  (recall that $\vec{H}$ is the mean curvature vector of the 
surface), Bartnik shows that
$ \lim_{r_{B} \rightarrow \infty} m=m_0$ with $m_0\in C^\infty(\mathbb{S}^2)$ and that,
under sufficient decay of suitable components of
the Einstein tensor which include the energy flux decay
condition of this paper,
\begin{align*}
E_B & = \underset{r_{B} \to
\infty}{\lim}\frac{1}{4\pi}\int_{\mathbb{S}^2}m
\volunitdos
=\frac{1}{4\pi}\int_{\mathbb{S}^2}m_0 \volunitdos,  \\
P_B^i & =\underset{r \to
\infty}{\lim}\frac{1}{4\pi}\int_{\mathbb{S}^2}m Y^1_i
\volunitdos =\frac{1}{4\pi}\int_{\mathbb{S}^2}m_0 Y^1_i \volunitdos.
\end{align*}
The null quasi-spherical gauge is such that 
$\theta_k - \frac{2}{r_{B}}$ is automatically a
divergence. Thus, for our results to fit with his it is necessary to select
$r_0$ in the geodesic background foliation $\{ S_{r}\}$ (which is of the form
$r = r_{B} + \xi$) so that 
$\c+2r_0$ is a divergence. An explicit computation shows that, in terms
of our notation, $m_0=\frac{1}{4}(\a-\c-4r_0)$. Bartnik's result 
is recovered from
(\ref{EuB})  because, with the shorthand
$\Psi_{u} := u_0 - u^i Y^1_{i}$,
\begin{align*}
E^u_B & =  
\frac{1}{16\pi}\int_{\mathbb{S}^2}\left(\triangle_\q\c+(\c+\a)
\right) \Psi_{u} \volunitdos \\
& = 
\frac{1}{16\pi}\int_{\mathbb{S}^2}\left(\triangle_\q\c+(\c+\a)
- (\triangle_\q+2)(\c+2r_0) \right) \Psi_{u} \volunitdos \\
& 
=\frac{1}{16\pi}\int_{\mathbb{S}^2}(\a-\c-4r_0) \Psi_{u} \volunitdos 
 =\frac{1}{4\pi}\int_{\mathbb{S}^2}m_0 \Psi_{u} \volunitdos
\end{align*}
where in the second expression we added zero 
in the form $0=  \int_{\mathbb{S}^2}\left(-(\triangle_\q+2)(\c+2r_0)\right) 
\Psi_{u}  \volunitdos$.
\end{remark}

\vspace{3mm}

\section*{Acknowledgments}
Financial support under the project  FIS2012-30926 (MICINN)
is acknowledged. A.S. acknowledges the Ph.D. grant AP2009-0063 (MEC).

\end{document}